\pdfoutput=1 

\documentclass[sigconf,nonacm]{aamas}

\usepackage{balance} 

\setcopyright{ifaamas}
\acmConference[AAMAS '21]{Proc.\@ of the 20th International Conference on Autonomous Agents and Multiagent Systems (AAMAS 2021)}{May 3--7, 2021}{Online}{U.~Endriss, A.~Now\'{e}, F.~Dignum, A.~Lomuscio (eds.)}
\copyrightyear{2021}
\acmYear{2021}
\acmDOI{}
\acmPrice{}
\acmISBN{}

\acmSubmissionID{393}

\title{Rankings for Bipartite Tournaments via Chain Editing}

\author{Joseph Singleton}
\affiliation{
    \institution{Cardiff University}
    \city{Cardiff}
}
\email{singletonj1@cardiff.ac.uk}

\author{Richard Booth}
\affiliation{%
    \institution{Cardiff University}
    \city{Cardiff}
}
\email{boothr2@cardiff.ac.uk}

\begin{abstract}

Ranking the participants of a tournament has applications in voting, paired
comparisons analysis, sports and other domains. In this paper we introduce
\emph{bipartite tournaments}, which model situations in which two different
kinds of entity compete indirectly via matches against players of the opposite
kind; examples include education (students/exam questions) and solo sports
(golfers/courses).
In particular, we look to find rankings via \emph{chain graphs},
which correspond to bipartite tournaments in which the sets of adversaries
defeated by the players on one side are nested with respect to set inclusion.
Tournaments of this form have a natural and appealing ranking associated with
them. We apply \emph{chain editing} -- finding the minimum number of edge
changes required to form a chain graph -- as a new mechanism
for tournament ranking. The properties of these rankings are investigated in a
probabilistic setting, where they arise as maximum likelihood estimators, and
through the axiomatic method of social choice theory.
Despite some nice properties, two problems remain: an important anonymity axiom
is violated, and chain editing is \complexityclass{NP}-hard. We address both
issues by relaxing the minimisation constraint in chain editing, and
characterise the resulting ranking methods via a greedy approximation
algorithm.

\end{abstract}

\usepackage{cleveref}
\usepackage{bm}
\usepackage[utf8x]{inputenc}
\usepackage[inline]{enumitem}
\usepackage{marginnote}
\usepackage{arydshln}
\usepackage{tikz}
\usepackage{subfig}

\makeatletter
\let\overlinewithoriginalheight\overline
\newcommand*\overlinewithlessheight[1]{{\mathpalette\overline@aux{#1}}}
\newcommand*\overline@aux[2]{
  \begingroup
    \count0=\fam 
    \setbox0=\hbox{$\m@th #1\fam=\count0 #2$}
    \@tempdima=.4\ht0
    \setbox0=\hbox{$\m@th #1\fam=\count0\overlinewithoriginalheight{#2}$}%
    \advance\@tempdima by .6\ht0
    \ht0=\@tempdima 
    \usebox0
  \endgroup%
}
\let\overline\overlinewithlessheight

\newcommand{\ch}{\mathcal{C}}
\newcommand{\minch}[1]{\operatorname{\mathcal{M}}\left({#1}\right)}
\newcommand{\minchmon}[1]{\operatorname{\mathcal{M}}_{\mathsf{mon}}\left({#1}\right)}

\newcommand{\mindist}[1]{m({#1})}
\renewcommand{\phi}{\varphi}
\newcommand{\K}{\mathcal{K}}
\newcommand{\N}{\mathbb{N}}
\newcommand{\R}{\mathbb{R}}
\newcommand{\ale}{\preceq}
\newcommand{\alt}{\prec}
\newcommand{\aeq}{\approx}
\newcommand{\asymb}{\mathcal{A}}
\newcommand{\bsymb}{\mathcal{B}}
\newcommand{\anle}{\leqslant^{\asymb}}
\newcommand{\aneq}{\approx^{\asymb}}
\newcommand{\anlt}{<^{\asymb}}
\newcommand{\bnle}{\leqslant^{\bsymb}}
\newcommand{\ble}{\sqsubseteq}
\newcommand{\blt}{\sqsubset}
\newcommand{\beq}{\approx}
\newcommand{\tr}{\top}
\newcommand{\dual}[1]{{\overline{#1}}}
\newcommand{\vect}{\operatorname{vec}}
\newcommand{\argmin}{\operatorname*{arg\ min}}
\newcommand{\argmax}{\operatorname*{arg\ max}}
\newcommand{\rs}{\upharpoonright}
\newcommand{\dotprod}{\bullet}

\newcommand{\symdiff}{\mathrel{\triangle}}
\newcommand{\phicount}{{\phi_{\mathsf{count}}}}
\newcommand{\phicardint}{{\phi_{\mathsf{CI}}}}

\newcommand{\tuple}[1]{\langle{#1}\rangle}
\newcommand{\complexityclass}[1]{\textsf{#1}}
\newcommand{\ranks}[1]{\mathsf{ranks}({#1})}
\renewcommand{\intop}[1]{{\phi_{#1}^\mathsf{int}}}
\newcommand{\swap}[3]{\mathsf{swap}({#1}; {#2}, {#3})}
\newcommand{\inlineheading}[1]{%
    \addvspace{1mm}%
    \noindent\textbf{#1.}
}

\graphicspath{ {./images/} }

\newlist{inlinelist}{enumerate*}{1}
\setlist[inlinelist]{label=(\roman*)}

\newcommand{\axiomref}[1]{\textbf{#1}}

\theoremstyle{remark} \newtheorem*{notation}{Notation}

\newtheoremstyle{axiomstyle}{}{}{\slshape}{}{\bfseries}{}{ }{\thmnote{ (#3) }}
\theoremstyle{axiomstyle} \newtheorem*{axiom}{}

\def\thisistheprerint{}

\begin{document}


\pagestyle{fancy}
\fancyhead{}


\maketitle


\section{Introduction}
\label{sec:introduction}

A tournament consists of a finite set of players equipped with a \emph{beating
relation} describing pairwise comparisons between each pair of players.
Determining a ranking of the players in a tournament has applications in voting
in social choice~\cite{brandt2016a} (where players represent alternatives and
$x$ beats $y$ if a majority of voters prefer $x$ over $y$), paired comparisons
analysis~\cite{gonzalez2014paired} (where players may represent products and
the beating relation the preferences of a user), search engines
\cite{slutzki2006scoring}, sports tournaments~\cite{bozoki2016application} and
other domains.

In this paper we introduce \emph{bipartite tournaments}, which consist of two
disjoint sets of players $A$ and $B$ such that comparisons only take place
between players from opposite sets. We consider ranking methods which produce
two rankings for each tournament -- one for each side of the bipartition. Such
tournaments model situations in which two different kinds of entity compete
\emph{indirectly} via matches against entities of the opposite kind.
The notion of competition may be abstract, which allows the model to be applied
in a variety of settings. An important example is
education~\cite{jiao2017algorithms}, where $A$ represents students, $B$ exam
questions, and student $a$ `beats' question $b$ by answering it correctly. Here
the ranking of students reflects their performance in the exam, and the ranking
of questions reflects their \emph{difficulty}. The simultaneous ranking of both
sides allows one ranking to influence the other; e.g. so that students are
rewarded for correctly answering difficult questions. This may prove
particularly useful in the context of crowdsourced questions provided by
students themselves, which may vary in their difficulty (see for example the
PeerWise system \cite{denny_peerwise_2008}).

A related example is \emph{truth discovery}~\cite{li_survey_2016,
singleton_booth_2020}: the task of finding true information on a number of
topics when faced with conflicting reports from sources of varying (but
unknown) reliability. Many truth discovery algorithms operate iteratively,
alternately estimating the reliability of sources based on current estimates of
the true information, and obtaining new estimates of the truth based on source
reliability levels. The former is an instance of a bipartite tournament;
similar to the education example, $A$ represents data sources, $B$ topics of
interest, and $a$ defeats $b$ by providing true information on topic $b$
(according to the current estimates of the truth). Applying a bipartite
tournament ranking method at this step may therefore facilitate development of
\emph{difficulty-aware} truth discovery algorithms, which reward sources for
providing accurate information on difficult topics
\cite{galland_corroborating_2010}.
Other application domains include the evaluation of generative models
in machine learning~\cite{olsson2018skill} (where $A$ represents generators and
$B$ discriminators) and solo sports contests (e.g. where $A$ represents golfers
and $B$ golf courses).

In principle, bipartite tournaments are a special case of \emph{generalised}
tournaments
\cite{gonzalez2014paired,slutzki2005ranking,csato2019impossibility}, which
allow intensities of victories and losses beyond a binary win or loss (thus
permitting draws or multiple comparisons), and drop the requirement that every
player is compared to all others.  However, many existing ranking methods in
the literature do not apply to bipartite tournaments due to the violation of an
\emph{irreducibility} requirement, which requires that the tournament graph be
strongly connected. In any case, bipartite tournament ranking presents a unique
problem -- since we aim to rank players with only indirect information
available -- which we believe is worthy of study in its own right.

In this work we focus particularly on ranking via \emph{chain graphs} and
\emph{chain editing}. A chain graph is a bipartite graph in which the
neighbourhoods of vertices on one side form a chain with respect to set
inclusion. A (bipartite) tournament of this form represents an `ideal'
situation in which the capabilities of the players are perfectly nested: weaker
players defeat a subset of the opponents that stronger players defeat. In this
case a natural ranking can be formed according to the set of opponents defeated
by each player. These rankings respect the tournament results in an intuitive
sense: if a player $a$ defeats $b$ and $b'$ ranks worse than $b$, then $a$ must
defeat $b'$ also.
Unfortunately, this perfect nesting may not hold in reality: a weak player may
win a difficult match by coincidence, and a strong player may lose a match by
accident.
With this in mind, \citet{jiao2017algorithms} suggested an appealing ranking
method for bipartite tournaments: apply \emph{chain editing} to the input
tournament -- i.e. find the minimum number of edge changes required to form a
chain graph -- and output the corresponding rankings. Whilst their work
focused on algorithms for chain editing and its variants, we look to study the
properties of the ranking method itself through the lens of computational social
choice.

\inlineheading{Contribution} Our primary contribution is the introduction of a
class of ranking mechanisms for bipartite tournaments defined by chain editing.
We also provide a new probabilistic characterisation of chain editing via
maximum likelihood estimation. To our knowledge this is the first in-depth
study of chain editing as a ranking mechanism. Secondly, we introduce a new
class of `chain-definable' mechanisms by relaxing the minimisation constraint
of chain editing in order to obtain tractable algorithms and to resolve the
failure of an important anonymity axiom.

\inlineheading{Paper outline} In \Cref{sec:preliminaries} we define the
framework for bipartite tournaments and introduce chain graphs.
\Cref{sec:ranking_via_chain_editing} outlines how one may use chain editing to
rank a tournament, and characterises the resulting mechanisms in a
probabilistic setting. Axiomatic properties are considered in
\Cref{sec:axiomatic_analysis}. \Cref{sec:match_preference_operators} defines a
concrete scheme for producing chain-editing-based rankings.
\Cref{sec:relaxing_chain_min} introduces new ranking methods by relaxing the
chain editing requirement. Related work is discussed in
\Cref{sec:related_work}, and we conclude in \Cref{sec:conclusion}.
\ifdefined\thisistheprerint
    Note that some proofs are omitted in the body of the paper and can be found
    in the appendix.
\else
    Note that some proofs are omitted due to lack of space, and can be found in
    the appendix of~\cite{singleton_booth_21_arxiv}.
\fi

\section{Preliminaries}
\label{sec:preliminaries}

In this section we define our framework for bipartite tournaments, introduce
chain graphs and discuss the link between them.

\subsection{Bipartite Tournaments}

Following the literature on generalised
tournaments~\cite{gonzalez2014paired, slutzki2005ranking,
csato2019impossibility}, we represent a tournament as a matrix, whose entries
represent the results of matches between participants. In what follows, $[n]$
denotes the set $\{1,\ldots,n\}$ whenever $n \in \N$.

\begin{definition}

    A \emph{bipartite tournament} -- hereafter simply a \emph{tournament} -- is
    a triple $(A, B, K)$, where $A = [m]$ and $B = [n]$ for some $m, n \in \N$,
    and $K$ is an $m \times n$ matrix with $K_{ab} \in \{0, 1\}$ for all $(a,
    b) \in A \times B$. The set of all tournaments will be denoted by $\K$.

\end{definition}

Here $A$ and $B$ represent the two sets of players in the
tournament.\footnotemark{} An entry $K_{ab}$ gives the result of the match
between $a \in A$ and $b \in B$: it is 1 if $a$ defeats $b$ and 0 otherwise.
Note that we do not allow for the possibility of draws, and every $a \in A$
faces every $b \in B$.
When there is no ambiguity we denote a tournament simply by $K$, with the
understanding that $A = [\text{rows}(K)]$ and $B = [\text{columns}(K)]$.

The \emph{neighbourhood} of a player $a \in A$ in $K$ is the set $K(a) = \{b
\in B \mid K_{ab} = 1\} \subseteq B$, i.e. the set of players which $a$
defeats. The neighbourhood of $b \in B$ is the set $K^{-1}(b) = \{a \in A \mid
K_{ab} = 1\} \subseteq A$, i.e. the set of players defeating $b$.

\footnotetext{
    Note that $A$ and $B$ are not disjoint as sets: $1$ is always contained in
    both $A$ and $B$, for instance. This poses no real problem, however, since
    we view the number $1$ merely a \emph{label} for a player. It will always
    be clear from context whether a given integer should be taken as a label
    for a player on the $A$ side or the $B$ side.
}

Given a tournament $K$, our goal is to place a ranking on each of $A$ and $B$. We
define a ranking \emph{operator} for this purpose.

\begin{definition}

    An \emph{operator} $\phi$ assigns each tournament $K$ a pair $\phi(K) =
    ({\ale_K^\phi}, {\ble_K^\phi})$ of total preorders on $A$ and $B$
    respectively.\footnotemark{}

\end{definition}

\footnotetext{
    A total preorder is a transitive and complete binary relation.
}

For $a, a' \in A$, we interpret $a \ale_K^\phi a'$ to mean that $a'$ is ranked
\emph{at least as strong} as $a$ in the tournament $K$, according to the
operator $\phi$ (similarly, $b \ble_K^\phi b'$ means $b'$ is ranked at least as
strong as $b$). The strict and symmetric parts of ${\ale_K^\phi}$ are denoted
by ${\alt_K^\phi}$ and ${\aeq_K^\phi}$,

As a simple example, consider $\phicount$, where $a \ale_K^\phicount a'$ iff
$|K(a)| \le |K(a')|$ and $b \ble_K^\phicount b'$ iff $|K^{-1}(b)| \ge
|K^{-1}(b')|$. This operator simply ranks players by number of victories. It is
a bipartite version of the \emph{points system} introduced
by~\citet{rubinstein1980ranking}, and generalises \emph{Copeland's
rule}~\cite{brandt2016a}.

\subsection{Chain Graphs}

Each bipartite tournament $K$ naturally corresponds to a bipartite graph $G_K$,
with vertices $A \sqcup B$ and an edge between $a$ and $b$ whenever $K_{ab} =
1$.\footnotemark{} The task of ranking a tournament admits a particularly
simple solution if this graph happens to be a \emph{chain graph}.

\footnotetext{
    $A \sqcup B$ is the \emph{disjoint union} of $A$ and $B$, which we define
    as $\{(a, \asymb) \mid a \in A\} \cup \{(b, \bsymb) \mid b \in B\}$, where
    $\asymb$ and $\bsymb$ are constant symbols.
}

\begin{definition}[\cite{yannakakis1981computing}]
\label{def:chain_graph}

    A bipartite graph $G = (U, V, E)$ is a \emph{chain graph} if there is an
    ordering $U = \{u_1,\ldots,u_k\}$ of $U$ such that $N(u_1) \subseteq \cdots
    \subseteq N(u_k)$, where $N(u_i) = \{v \in V \mid (u_i, v) \in E\}$ is the
    neighbourhood of $u_i$ in $G$.

\end{definition}

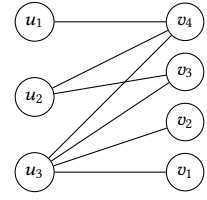
\begin{figure}
    \centering
    \begin{tikzpicture}

        \def \m {3}
        \def \n {4}
        \def \width {2}
        \def \height {2}
        \def \edges {{{1,1},{1,2}}}

        \tikzset{mynode/.style={circle,draw}}

        \footnotesize

        \foreach \a in {1,...,\m} {
            \def \y {\height * (\a - 1) / (\m - 1)}
            \node (a-\a) at (0, -{\y}) [mynode] {$u_{\a}$};
        }
        \foreach \b in {1,...,\n} {
            \def \y {\height * (\b - 1) / (\n - 1)}
            \def \index {\pgfmathparse{int(1 + \n - \b)}\pgfmathresult}
            \node (b-\b) at (\width, -{\y}) [mynode] {$v_{\index}$};
        }
        \def \drawedge#1#2{\draw[-] (a-#1) -- (b-#2);}
        \drawedge{1}{1}
        \drawedge{2}{1} \drawedge{2}{2}
        \drawedge{3}{1} \drawedge{3}{2} \drawedge{3}{3} \drawedge{3}{4}

    \end{tikzpicture}
    \caption{An example of a chain graph}
    \label{fig:chain_graph_example}
    \Description{An example of a chain graph}
\end{figure}

In other words, a chain graph is a bipartite graph where the neighbourhoods of
the vertices on one side can be ordered so as to form a chain with respect to
set inclusion. It is easily seen that this nesting property holds for $U$ if
and only if it holds for $V$. \Cref{fig:chain_graph_example} shows an example
of a chain graph.

Now, as our terminology might suggest, the neighbourhood $K(a)$ of some player
$a \in A$ in a tournament $K$ coincides with the neighbourhood of the
corresponding vertex in $G_K$. If $G_K$ is a chain graph we can therefore
enumerate $A$ as $\{a_1,\ldots,a_m\}$ such that $K(a_i) \subseteq K(a_{i+1})$
for each $1 \le i < m$. This indicates that each $a_{i+1}$ has performed
\emph{at least as well} as $a_i$ in a strong sense: every opponent which $a_i$
defeated was also defeated by $a_{i+1}$, and $a_{i+1}$ may have additionally
defeated opponents which $a_i$ did not.\footnotemark{} It seems only natural in
this case that one should rank $a_i$ (weakly) below $a_{i+1}$.
\footnotetext{
    Note that this is a more robust notion of performance than comparing the
    neighbourhoods of $a_i$ and $a_{i+1}$ by \emph{cardinality}, which may fail
    to account for differences in the strength of opponents when counting wins
    and losses.
}
Appealing to transitivity and the fact that each $a \in A$ appears as
\emph{some} $a_i$, we see that any tournament $K$ where $G_K$ is a chain graph
comes pre-equipped with a natural total preorder on $A$, where $a'$ ranks
higher than than $a$ if and only if $K(a) \subseteq K(a')$. The duality of the
neighbourhood-nesting property for chain graphs implies that $B$ can also be
totally preordered, with $b'$ ranked higher than $b$ if and only if $K^{-1}(b)
\supseteq K^{-1}(b')$.\footnotemark{}
Moreover, these total preorders relate to the tournament results in an
important sense: if $a$ defeats $b$ and $b'$ ranks worse than $b$, then $a$
must defeat $b'$ also. That is, the neighbourhood of each $a \in A$ is
\emph{downwards closed} w.r.t the ranking of $B$, and the neighbourhood of each
$b \in B$ is \emph{upwards closed} in $A$.

\footnotetext{
    Note that the ordering of the $B$s is reversed compared to the $A$s, since
    the larger $K^{-1}(b)$ the \emph{worse} $b$ has performed.
}

Tournaments corresponding to chain graphs will be said to satisfy the
\emph{chain property}, and will accordingly be called \emph{chain tournaments}.
We give a simpler (but equivalent) definition which does not refer to the
underlying graph $G_K$. First, define relations ${\anle_K}, {\bnle_K}$ on
$A$ and $B$ respectively by $a \anle_K a'$ iff $K(a) \subseteq K(a')$ and $b
\bnle_K b'$ iff $K^{-1}(b) \supseteq K^{-1}(b')$, for any tournament $K$.

\begin{definition}
    A tournament $K$ has the \emph{chain property} if $\anle_K$ is a total
    preorder.
\end{definition}

According to the duality principle mentioned already, the chain property
implies that $\bnle_K$ is also a total preorder. Note that the relations
$\anle_K$ and $\bnle_K$ are analogues of the \emph{covering relation} for
non-bipartite tournaments \cite{brandt2016a}.

\begin{example}
    Consider
    $
        K = \left[\begin{smallmatrix}
            1 & 0 & 0 & 0 \\
            1 & 1 & 0 & 0 \\
            1 & 1 & 1 & 1
        \end{smallmatrix}\right]
    $. Then $K(1) \subset K(2) \subset K(3)$, so $K$ has the chain
    property.  In fact, $K$ is the tournament corresponding to the chain graph
    $G$ from \Cref{fig:chain_graph_example}.

\end{example}

\section{Ranking via Chain Editing}
\label{sec:ranking_via_chain_editing}

We have seen that chain tournaments
come equipped with natural rankings of $A$ and $B$. Such tournaments represent
an `ideal' situation, wherein the abilities of the players on both sides of the
tournament are perfectly nested. Of course this may not be so in reality:
the nesting may be broken by some $a \in A$ winning a match it ought not to by
chance, or by losing a match by accident.

One idea for recovering a ranking in this case, originally suggested
by \citet{jiao2017algorithms}, is to apply \emph{chain editing}: find the
minimum number of edge changes required to convert the graph $G_K$ into a chain
graph. This process can be seen as correcting the `noise' in an observed
tournament $K$ to obtain an ideal ranking. In this section we introduce the
class of operators producing rankings in this way.

\subsection{Chain-minimal Operators}

To define chain-editing in our framework we once again present an equivalent
definition which does not refer to the underlying graph $G_K$: the number of
edge changes between graphs can be replaced by the \emph{Hamming distance}
between tournament matrices.

\begin{definition}
    For $m, n \in \N$, let $\ch_{m,n}$ denote the set of all $m \times n$
    chain tournaments. For an $m \times n$ tournament $K$,
    write $\minch{K} = \argmin_{K' \in \ch_{m,n}}{d(K, K')} \subseteq \K$ for
    the set of chain tournaments closest to $K$ w.r.t the Hamming
    distance $d(K, K') = |\{(a,b) \in A \times B \mid K_{ab} \ne K'_{ab}\}|$.
    Let $\mindist{K}$ denote this minimum distance.
\end{definition}

Note that chain editing, which is \complexityclass{NP}-hard in
general~\cite{jiao2017algorithms}, amounts to finding a single element of
$\minch{K}$.\footnotemark{} We comment further on the computational complexity
of chain editing in \Cref{sec:related_work}. The following
property characterises chain editing-based operators $\phi$.

\footnotetext{
    The decision problem associated with chain editing -- which in tournament
    terms is the question of whether $\mindist{K} \le k$ for a given integer
    $k$ -- is \complexityclass{NP}-complete~\cite{drange2015threshold}.
}

\begin{axiom}[chain-min]
    For every tournament $K$ there is $K' \in \minch{K}$ such that $\phi(K) =
    ({\anle_{K'}}, {\bnle_{K'}})$.
\end{axiom}

\sloppy
That is, the ranking of $K$ is obtained by choosing the neighbourhood-subset
rankings for some closest chain tournament $K'$. Operators satisfying
\axiomref{chain-min} will be called \emph{chain-minimal}.

\begin{example}
    \label{ex:minch}
    Consider
    $
        K = \left[\begin{smallmatrix}
            1 & 0 & 1 & 0 \\
            1 & 1 & 0 & 0 \\
            0 & 1 & 1 & 1
        \end{smallmatrix}\right]
    $.
    $K$ does not have the chain property, since neither $K(1) \subseteq K(2)$
    nor $K(2) \subseteq K(1)$. The set $\minch{K}$ consists of four
    tournaments a distance of 2 from $K$:
    \[
        \minch{K} = \left\{
            \left[\begin{smallmatrix}
                1 & \bm{{\color{red}1}} & 1 & 0 \\
                1 & 1 & 0 & 0 \\
                \bm{{\color{red}1}} & 1 & 1 & 1
            \end{smallmatrix}\right],
            \left[\begin{smallmatrix}
                1 & 0 & \bm{{\color{red}0}} & 0 \\
                1 & 1 & 0 & 0 \\
                \bm{{\color{red}1}} & 1 & 1 & 1
            \end{smallmatrix}\right],
            \left[\begin{smallmatrix}
                1 & 0 & 1 & 0 \\
                1 & \bm{{\color{red}0}} & 0 & 0 \\
                \bm{{\color{red}1}} & 1 & 1 & 1
            \end{smallmatrix}\right],
            \left[\begin{smallmatrix}
                1 & 0 & 1 & 0 \\
                1 & 1 & \bm{{\color{red}1}} & 0 \\
                \bm{{\color{red}1}} & 1 & 1 & 1
            \end{smallmatrix}\right]
        \right\}
    \]

    \sloppy
    The corresponding rankings are $(213, \{12\}34)$, $(123, 12\{34\})$, $(213,
    13\{24\})$ and $(123, \{13\}24)$.\footnotemark{}


    \footnotetext{
        Here $a_1a_2a_3$ is shorthand for the ranking $a_1 \alt a_2 \alt a_3$
        of $A$, and similar for $B$. Elements in brackets are ranked equally.
    }

\end{example}

\Cref{ex:minch} shows that there is no unique chain-minimal operator, since for
a given tournament $K$ there may be several closest chain tournaments to choose
from. In \Cref{sec:match_preference_operators} we introduce a principled way to
single out a \emph{unique} chain tournament and thereby construct a
well-defined chain-minimal operator.

\subsection{A Maximum Likelihood Interpretation}
\label{sec:mle}

So far we have motivated \axiomref{chain-min} as a way to fix
errors in a tournament and recover the ideal or \emph{true} ranking. In this
section we make this notion precise by defining a probabilistic model in which
chain-minimal rankings arise as maximum likelihood estimates.
The maximum likelihood approach has been applied for (non-bipartite)
tournaments (e.g. the Bradley-Terry
model~\cite{bradley_terry_52,gonzalez2014paired}), voting in social choice
theory~\cite{elkind2016rationalizations}, truth
discovery~\cite{wang_truth_2012}, belief merging~\cite{everaere2020} and other
related problems.

In this approach we take an epistemic view of tournament ranking: it is assumed
there exists a true `state of the world' which determines the tournament
results along with objective rankings of $A$ and $B$. A given
tournament $K$ is then seen as a \emph{noisy observation} derived from the
true state, and a \emph{maximum likelihood estimate} is a state for which the
probability of observing $K$ is maximal.

More specifically, a state of the world is represented as a vector of
\emph{skill levels} for the players in $A$ and $B$.\footnotemark{}

\footnotetext{
    For simplicity we use numerical skill levels here, although it would
   suffice to have a partial preorder on $A \sqcup B$ such that each
   $a \in A$ is comparable with every $b \in B$.
}

\begin{definition}
   \label{def:stateworld}

    For a fixed size $m \times n$, a \emph{state of the world} is a tuple
    $\theta = \tuple{\bm{x}, \bm{y}}$, where $\bm{x} \in \R^m$ and
    $\bm{y} \in \R^n$ satisfies the following properties:
   \begin{equation}
        \forall a, a' \in A \quad (
            x_a < x_{a'} \implies \exists b \in B: x_a < y_b \le x_{a'}
        )
        \label{eqn:state_condition_a}
   \end{equation}
   \begin{equation}
        \forall b, b' \in B \quad(
            y_b < y_{b'} \implies \exists a \in A: y_b \le x_a < y_{b'}
        )
        \label{eqn:state_condition_b}
   \end{equation}
   where $A = [m]$, $B = [n]$. Write $\Theta_{m,n}$ for the set of all $m
   \times n$ states.

\end{definition}

For $a \in A$, $x_a$ is the \emph{skill level} of $a$ in state $\theta$ (and
similarly for $y_b$). These skill levels represent the true capabilities of the
players in $A$ and $B$ in state $\theta$: $a$ is capable of defeating $b$ if
and only if $x_a \ge y_b$.
Note that \labelcref{eqn:state_condition_a} suggests a simple form of
\emph{explainability}: $a'$ can only be strictly more skilful than $a$ if there
is some $b \in B$ which \emph{explains} this fact, i.e. some $b$ which $a'$ can
defeat but $a$ cannot (\labelcref{eqn:state_condition_b} is analogous for the
$B$s). These conditions are intuitive if we assume that skill levels are
relative to the sets $A$ and $B$ currently under consideration (i.e.
they do not reflect the abilities of players in future matches against new
contenders outside of $A$ or $B$). Finally note that our states of the world
are \emph{richer} than the output of an operator, in contrast to other work in
the literature~\cite{bradley_terry_52, gonzalez2014paired,
elkind2016rationalizations}. Specifically, a state $\theta$ contains extra
information in the form of comparisons between $A$ and $B$.

Noise is introduced in the observed tournament $K$ via \emph{false positives}
(where $a \in A$ defeats a more skilled $b \in B$ by accident) and \emph{false
negatives} (where $a \in A$ is defeated by an inferior $b \in B$ by
mistake).\footnote{Note that a false positive for $a$ is a false negative for
$b$ and vice versa.} The noise model is therefore parametrised by the false
positive and false negative rates $\bm{\alpha} = \tuple{\alpha_+, \alpha_-}
\in [0,1]^2$, which we assume are the same for all $a \in A$.\footnotemark{} We
also assume that noise occurs independently across all matches.

\footnotetext{
    This is a strong assumption, and it may be more realistic to model the
    false positive/negative rates as a function of $x_a$. We leave this to
    future work.
}

\begin{definition}
   \label{def:probdist}

   Let $\bm{\alpha} = \tuple{\alpha_+, \alpha_-} \in [0,1]^2$. For each $m,
   n \in \N$ and $\theta = \tuple{\bm{x}, \bm{y}} \in \Theta_{m,n}$, consider
   independent binary random variables $X_{ab}$ representing the outcome of a
   match between $a \in [m]$ and $b \in [n]$, where
   \begin{equation}
        \label{eqn:probdist_random_var_one}
        P_{\bm{\alpha}}(X_{ab} = 1 \mid \theta)
        = \begin{cases}
            \alpha_+,& x_a < y_b \\
            1 - \alpha_-,& x_a \ge y_b
        \end{cases}
   \end{equation}
   \begin{equation}
        \label{eqn:probdist_random_var_zero}
        P_{\bm{\alpha}}(X_{ab} = 0 \mid \theta)
        = \begin{cases}
            1 - \alpha_+,& x_a < y_b \\
            \alpha_-,& x_a \ge y_b
        \end{cases}
   \end{equation}

   This defines a probability distribution $P_{\bm{\alpha}}({\cdot} \mid
   \theta)$ over $m \times n$ tournaments by
   \[
      P_{\bm{\alpha}}(K \mid \theta) = \prod_{(a, b) \in [m] \times [n]}{
           P_{\bm{\alpha}}(X_{ab} = K_{ab} \mid \theta)
      }
   \]
\end{definition}

Here $P_{\bm{\alpha}}(K \mid \theta)$ is the probability of observing the
tournament results $K$ when the false positive and negative rates are given by
$\bm{\alpha}$ and the true state of the world is $\theta$. Note that the
four cases in \labelcref{eqn:probdist_random_var_one} and
\labelcref{eqn:probdist_random_var_zero} correspond to a false positive, true
positive, true negative and false negative respectively. We can now define a
maximum likelihood operator.

\begin{definition}

    Let $\bm{\alpha} \in [0,1]^2$ and $m, n \in \N$. Then $\theta \in
    \Theta_{m,n}$ is a \emph{maximum likelihood estimate} (MLE) for an $m
    \times n$ tournament $K$ w.r.t $\bm{\alpha}$ if $\theta \in
    \argmax_{\theta' \in \Theta_{m,n}}{P_{\bm{\alpha}}(K \mid \theta')}$. An
    operator $\phi$ is a \emph{maximum likelihood operator} w.r.t
    $\bm{\alpha}$ if for any $m, n \in \N$ and any $m \times n$ tournament $K$
    there is an MLE $\theta = \tuple{\bm{x}, \bm{y}} \in \Theta_{m,n}$ for $K$
    such that $a \ale_K^\phi a'$ iff $x_a \le x_{a'}$ and $b \ble_K^\phi b'$
    iff $y_b \le y_{b'}$.

\end{definition}

Now, consider the tournament $K_\theta$ associated with each state $\theta =
\tuple{\bm{x}, \bm{y}}$, given by $[K_\theta]_{ab} = 1$ if $x_a \ge y_b$ and
$[K_\theta]_{ab} = 0$ otherwise. Note that $K_\theta$ is the unique tournament
with non-zero probability when
there are no false positive or false negatives. Expressed in terms of
$K_\theta$, the MLEs take a particularly simple form if $\alpha_+ = \alpha_-$,
i.e. if false positives and false negatives occur at the same rate.

\begin{lemma}
   \label{result:mle_hamming}

   Let $\bm{\alpha} = \tuple{\beta, \beta}$ for some $\beta < \frac{1}{2}$.
   Then $\theta$ is an MLE for $K$ if and only if $\theta \in \argmin_{\theta'
   \in \Theta_{m,n}}{d(K, K_{\theta'})}$.
\end{lemma}

\begin{proof}[Proof (sketch)]

    Let $K$ be an $m \times n$ tournament. It can be shown (and we do so in the
    appendix) that for any $\theta \in \Theta_{m,n}$
    \begin{equation*}
       \begin{split}
           P_{\bm{\alpha}}(K \mid \theta)
           =
           \Big(
           \prod_{a \in A}
             &\alpha_+^{|K(a) \setminus K_\theta(a)|}
             (1 - \alpha_-)^{|K(a) \cap K_\theta(a)|}
             \\
             &\quad (1 - \alpha_+)^{|B \setminus (K(a) \cup K_\theta(a))|}
             \alpha_-^{|K_\theta(a) \setminus K(a)|}
           \Big)
       \end{split}
    \end{equation*}
    Plugging in $\alpha_+ = \alpha_- = \beta$ and simplifying, one can obtain
    \[
       \begin{aligned}
       P_{\bm{\alpha}}(K \mid \theta)
           &=  c
               \prod_{a \in A}{
                \left(
                  \frac{\beta}{1 - \beta}
                \right)^{
                  |K(a) \symdiff K_\theta(a)|
                }
           }
       \end{aligned}
    \]
    where $X \symdiff Y = (X \setminus Y) \cup (Y \setminus X)$ is the
    symmetric difference of two sets $X$ and $Y$, and $c = (1
    - \beta)^{|A| \cdot |B|}$ is a positive constant that does not depend on
    $\theta$. Now, $P_{\bm{\alpha}}(K \mid \theta)$ is positive, and is maximal
    when its logarithm is. We have
    \[
       \begin{aligned}
           \log{P_{\bm{\alpha}}(K \mid \theta)}
           &=
              \log{c}
              +
              \log{\left(\frac{\beta}{1 - \beta}\right)}
              \sum_{a \in A}{
                  |K(a) \symdiff K_\theta(a)|
              } \\
           &=
              \log{c}
              +
              \log{\left(\frac{\beta}{1 - \beta}\right)}
              d(K, K_\theta)
       \end{aligned}
    \]
    Since $\log{c}$ is constant and $\beta < 1/2$ implies
    $\log{\left(\frac{\beta}{1 - \beta}\right)} < 0$, it follows that
    $\log{P_{\bm{\alpha}}(K \mid \theta)}$ is maximised exactly when $d(K,
    K_\theta)$ is minimised, which proves the result.
\end{proof}

This result characterises the MLE states for $K$ as those for which $K_\theta$
is the closest to $K$. As it turns out, the tournaments $K_\theta$ that arise
in this way are exactly those with the chain property.

\begin{lemma}
   \label{result:chain_iff_ktheta}

   An $m \times n$ tournament $K$ has the chain property if and only if $K =
   K_\theta$ for some $\theta \in \Theta_{m,n}$.

\end{lemma}

The proof of \Cref{result:chain_iff_ktheta} relies crucially on
\labelcref{eqn:state_condition_a} and \labelcref{eqn:state_condition_b} in the
definition of a state. Combining all the results so far we obtain our first
main result: the maximum likelihood operators for $\bm{\alpha} =
\tuple{\beta,\beta}$ are exactly the chain-minimal operators.

\begin{theorem}
   \label{result:mle_iff_chainmin_operator}
   Let $\bm{\alpha} = \tuple{\beta, \beta}$ for some $\beta < \frac{1}{2}$.
   Then $\phi$ is a maximum likelihood operator w.r.t $\bm{\alpha}$
   if and only if $\phi$ satisfies \axiomref{chain-min}.
\end{theorem}

\begin{proof}[Proof (sketch)]

    First note that by \Cref{result:mle_hamming}, a state $\theta$ is an MLE
    for an $m \times n$ tournament $K$ iff $K_\theta$ is closest to $K$ amongst
    all other tournaments $\{K_{\theta'} \mid \theta' \in \Theta_{m,n}\}$. But
    by \Cref{result:chain_iff_ktheta}, this set is exactly the $m \times n$
    tournaments with the chain property. It follows from the definition of
    $\minch{K}$ that $\theta$ is an MLE if and only if $K_\theta \in
    \minch{K}$. Consequently, $K' \in \minch{K}$ if and only if $K' = K_\theta$
    for some MLE $\theta$ for $K$.
    We see that \axiomref{chain-min} can be
    equivalently stated as follows: for all $K$ there exists an MLE
    $\theta$ such that $\phi(K) = ({\anle_{K_\theta}}, {\bnle_{K_\theta}})$.
    Using properties \labelcref{eqn:state_condition_a} and
    \labelcref{eqn:state_condition_b} in \Cref{def:stateworld} for $\theta$ it
    is straightforward to show that $a \anle_{K_\theta} a'$ iff $x_a \le
    x_{a'}$ and $b \bnle_{K_\theta} b'$ iff $y_b \le y_{b'}$ for all $a, a' \in
    A$, $b, b' \in B$ (where $\theta = \tuple{\bm{x}, \bm{y}}$). This means
    that the above reformulation of \axiomref{chain-min} coincides with the
    definition of a maximum likelihood operator, and we are done.
\end{proof}

Similar results can be obtained for
other limiting values of $\bm{\alpha}$. If $\alpha_+ = 0$ and $\alpha_- \in (0,
1)$ then the MLE operators correspond to \emph{chain completion}: finding
the minimum number of edge \emph{additions} required to make $G_K$ a chain graph. This
models situations where false positives never occur, although false negatives
may (e.g. numerical entry questions in the case where $A$ represents students
and $B$ exam questions~\cite{jiao2017algorithms}). Similarly, the case
$\alpha_- = 0$ and $\alpha_+ \in (0, 1)$ corresponds to \emph{chain deletion},
where edge additions are not allowed.

\section{Axiomatic analysis}
\label{sec:axiomatic_analysis}

Chain-minimal operators have theoretical backing in a probabilistic sense due
to the results of \Cref{sec:mle}, but are they appropriate ranking methods in
practise? To address this question we consider the \emph{normative} properties
of chain-minimal operators via the axiomatic method of social choice theory. We
formulate several axioms for bipartite tournament ranking
and assess whether they
are compatible with \axiomref{chain-min}. It will be seen that an important
\emph{anonymity} axiom fails for all chain-minimal operators; later in
\Cref{sec:match_preference_operators} we describe a scenario in which this is
acceptable and define a class of concrete operators for this case, and in
\Cref{sec:relaxing_chain_min} we relax the \axiomref{chain-min} requirement in
order to gain anonymity.

\subsection{The Axioms}

We will consider five axioms -- mainly adaptations of standard social choice
properties to the bipartite tournament setting.

\inlineheading{Symmetry Properties}
We consider two symmetry properties. The first is a classic \emph{anonymity}
axiom, which says that an operator $\phi$ should not be sensitive to the
`labels' used to identify participants in a tournament. Axioms of this form are
standard in social choice theory; a tournament version goes at least as far
back as~\cite{rubinstein1980ranking}.

We need some notation: for a tournament $K$ and permutations $\sigma: A \to A$,
$\pi: B \to B$, let $\sigma(K)$ and $\pi(K)$ denote the tournament obtained by
permuting the rows and columns of $K$ by $\sigma$ and $\pi$ respectively, i.e.
$[\sigma(K)]_{ab} = K_{\sigma^{-1}(a), b}$ and $[\pi(K)]_{ab} = K_{a,
\pi^{-1}(b)}$. Note that in the statement of the axioms we omit universal
quantification over $K$, $a, a' \in A$ and $b, b' \in B$ for
brevity.

\begin{axiom}[anon]
    Let $\sigma:A \to A$ and $\pi:B \to B$ be permutations. Then $a \ale_K^\phi
    a'$ iff $\sigma(a) \ale_{\pi(\sigma(K))}^\phi \sigma(a')$.
\end{axiom}

Our second axiom is specific to bipartite tournaments, and expresses a
\emph{duality} between the two sides $A$ and $B$: given the two sets of
conceptually disjoint entities participating in a bipartite tournament, it
should not matter which one we label $A$ and which one we label $B$. We need
the notion of a \emph{dual tournament}.

\begin{definition}

    The \emph{dual tournament} of $K$ is $\dual{K} = \bm{1} - K^\tr$, where
    $\bm{1}$ denotes the matrix consisting entirely of 1s.

\end{definition}

$\dual{K}$ is essentially the same tournament as $K$, but with the roles of $A$
and $B$ swapped. In particular, $A_K = B_\dual{K}$, $B_K = A_\dual{K}$ and
$K_{ab} = 1$ iff $\dual{K}_{ba} = 0$. Also note that $\dual{\dual{K}} = K$.
The duality axiom states that the ranking of the $B$s in $K$ is the same as the
$A$s in $\dual{K}$.

\begin{axiom}[dual]
    $b \ble_K^\phi b'$ iff $b \ale_\dual{K}^\phi b'$.
\end{axiom}

Whilst \axiomref{dual} is not necessarily a universally desirable property --
one can imagine situations where $A$ and $B$ are not fully abstract and should
not be treated symmetrically -- it is important to consider in any study of
bipartite tournaments. Note that \axiomref{dual} implies $a \ale_K^\phi
a'$ iff $a \ble_{\dual{K}}^\phi a'$, so that a \axiomref{dual}-operator can be
defined by giving the ranking for one of $A$ or $B$ only, and defining the
other by duality. This explains our choice to define \axiomref{anon} (and
subsequent axioms) solely in terms of the $A$ ranking: the analogous anonymity
constraint for the $B$ ranking follows from \axiomref{anon} together with
\axiomref{dual}.

\inlineheading{An Independence Property}
\emph{Independence axioms} play a crucial role in social choice. We present a
bipartite adaptation of a classic axiom introduced
in~\cite{rubinstein1980ranking}, which has subsequently been called
\emph{Independence of Irrelevant Matches}~\cite{gonzalez2014paired}.

\begin{axiom}[IIM]

    If $K_1, K_2$ are tournaments of the same size with identical $a$-th and
    $a'$-th rows, then $a \ale_{K_1}^\phi a'$ iff $a \ale_{K_2}^\phi a'$.

\end{axiom}

\axiomref{IIM} is a strong property, which says the relative ranking of $a$ and
$a'$ does not depend on the results of any match not involving $a$ or $a'$.
This axiom has been questioned for generalised
tournaments~\cite{gonzalez2014paired}, and a similar argument can be made
against it here: although each player in $A$ faces the same opponents, we may
wish to take the \emph{strength} of opponents into account, e.g. by rewarding
victories against highly-ranked players in $B$. Consequently we do not view
\axiomref{IIM} as an essential requirement, but rather introduce it to
facilitate comparison with our work and the existing tournament literature.

\inlineheading{Monotonicity Properties}
Our final axioms are monotonicity properties, which express the idea that
\emph{more victories are better}. The first axiom follows our original
intuition for constructing the natural ranking associated with a chain graph;
namely that $K(a) \subseteq K(a')$ indicates $a'$ has performed at least as
well as $a$.

\begin{axiom}[mon]
    If $K(a) \subseteq K(a')$ then $a \ale_K^\phi a'$.
\end{axiom}

Note that \axiomref{mon} simply says ${\ale_K^\phi}$ extends the (in general,
partial) preorder ${\anle_K}$.
Yet another standard axiom is \emph{positive responsiveness}.

\begin{axiom}[pos-resp]

    If $a \ale_K^\phi a'$ and $K_{a',b} = 0$ for some $b \in B$, then $a
    \alt_{K + \bm{1}_{a', b}}^\phi a'$, where $\bm{1}_{a', b}$ is the matrix
    with 1 in position $(a', b)$ and zeros elsewhere.

\end{axiom}

That is, adding an extra victory for $a$ should only improve its ranking, with
ties now broken in its favour. This version of positive responsiveness was
again introduced in~\cite{rubinstein1980ranking}, where together with
\axiomref{anon} and \axiomref{IIM} it characterises the \emph{points system}
ranking method for round-robin tournaments, which simply ranks players
according to the number of victories. The analogous operator in our framework
is $\phicount$, and it can be shown that $\phicount$ is uniquely characterised
by \axiomref{anon}, \axiomref{IIM}, \axiomref{pos-resp} and \axiomref{dual}.
Finally, note that \axiomref{pos-resp} also acts as a kind of
\emph{strategyproofness}: $a$ cannot improve its ranking by deliberately losing
a match. Specifically, if $K_{ab} = 1$ and $a \ale_K^\phi a'$, then
\axiomref{pos-resp} implies $a \alt_{K - \bm{1}_{ab}}^\phi a'$.

\subsection{Axiom Compatibility with \axiomref{chain-min}}


We come to analysing the compatibility of \axiomref{chain-min} with the axioms.
First, the negative results.

\begin{theorem}
    \label{result:chainmin_axiom_incompatibilities}

    There is no operator satisfying \axiomref{chain-min} and any of
    \axiomref{anon}, \axiomref{IIM} or \axiomref{pos-resp}.

\end{theorem}

The counterexample for \axiomref{anon} is particularly simple: take $K =
\left[\begin{smallmatrix} 1&0\\0&1 \end{smallmatrix}\right]$. Swapping the rows
and columns brings us back to $K$, so \axiomref{anon} implies $1, 2 \in A$ rank
equally. However, it is easily seen for every $K' \in \minch{K}$, either $K(1)
\subset K(2)$ or $K(2) \subset K(1)$, i.e no chain-minimal operator can rank 1
and 2 equally.

The MLE results of \Cref{sec:mle} provides informal
explanation for this result. For $K$ above to arise in the noise model of
\Cref{def:probdist} there must have been two `mistakes' (false positives or
false negatives). This is less likely than a single mistake from just one of
$1,2 \in A$, but the likelihood maximisation forces us to choose one or the
other. A similar argument explains the \axiomref{pos-resp} failure.

It is also worth noting that \axiomref{anon} only fails at the last step of
chain editing, where a single element of $\minch{K}$ is chosen. Indeed, the set
$\minch{K}$ itself \emph{does} exhibit the kind of symmetry one might expect:
we have $\minch{\pi(\sigma(K))} = \{\pi(\sigma(K')) \mid K' \in \minch{K}\}$.
This means that an operator which aggregates the rankings from \emph{all} $K'
\in \minch{K}$ -- e.g. any anonymous social welfare function --
would satisfy \axiomref{anon}. The other axioms are compatible with
\axiomref{chain-min}.

\begin{theorem}
    \label{result:chainmin_axiom_compatibilities}

    For each of \axiomref{dual} and \axiomref{mon}, there exists an operator
    satisfying \axiomref{chain-min} and the stated property.

\end{theorem}

Despite the simplicity of \axiomref{mon},
\Cref{result:chainmin_axiom_compatibilities} is deceptively difficult to prove.
We describe operators satisfying \axiomref{chain-def} and \axiomref{dual} or \axiomref{mon}
non-constructively by first taking an \emph{arbitrary} chain-minimal operator
$\phi$, and using properties of the set $\minch{K}$ to produce $\phi'$
satisfying \axiomref{dual} or \axiomref{mon}. Note also that we have not yet
constructed an operator satisfying \axiomref{dual}, \axiomref{mon} and
\axiomref{chain-min} simultaneously, although we conjecture that such operators
do exist.

\section{Match-preference operators}
\label{sec:match_preference_operators}


The counterexample for \axiomref{chain-min} and \axiomref{anon} suggests that
chain-minimal operators require some form of tie-breaking mechanism when the
tournaments in $\minch{K}$ cannot be distinguished while respecting anonymity.
While this limits the use of chain-minimal operators as general purpose ranking
methods, it is not such a problem if additional information is available to
guide the tie-breaking. In this section we introduce a new class of operators
for this case.

The core idea is to single out a unique chain tournament close to $K$ by paying
attention to not only the \emph{number} of entries in $K$ that need to be
changed to produce a chain tournament, \emph{which} entries. Specifically, we
assume the availability of a total order on the set of matrix indices $\N
\times \N$ (the \emph{matches}) which indicates our willingness to change an
entry in $K$: the higher up $(a, b)$ is in the ranking, the more acceptable it
is to change $K_{ab}$ during chain editing.

This total order -- called the \emph{match-preference relation} -- is fixed for
all tournaments $K$; this means we are dealing with extra information about how
tournaments are \emph{constructed in matrix form}, not extra information about
any specific tournament $K$.

One possible motivation for such a ranking comes from cases where matches occur
at distinct points in time. In this case the matches occurring more recently
are (presumably) more representative of the players' \emph{current} abilities,
and we should therefore prefer to modify the outcome of old matches where
possible.

For the formal definition we need notation for the \emph{vectorisation} of a
tournament $K$: for a total order ${\trianglelefteq}$ on $\N \times \N$ and an
$m \times n$ tournament $K$, we write $\vect_{\trianglelefteq}(K)$ for the
vector in $\{0,1\}^{mn}$ obtained by collecting the entries of $K$ in the order
given by ${\trianglelefteq} \rs (A \times B)$,\footnotemark{} starting with the
minimal entry. That is, $\vect_{\trianglelefteq}(K) = (K_{a_1,b_1}, \ldots,
K_{a_{mn},b_{mn}})$, where $(a_1,b_1), \ldots, (a_{mn},b_{mn})$ is the unique
enumeration of $A \times B$ such that $(a_i,b_i) \trianglelefteq
(a_{i+1},b_{i+1})$ for each $i$.

\footnotetext{
    This denotes the restriction of ${\trianglelefteq}$ to $A \times B$, i.e.
    ${\trianglelefteq} \cap ((A \times B) \times (A \times B))$.
}

The operator corresponding to $\trianglelefteq$ is defined using the notion of
a \emph{choice function}: a function $\alpha$ which maps any tournament $K$ to
an element of $\minch{K}$. Any such function defines a chain-minimal operator
$\phi$ by setting $\phi(K) = ({\anle_{\alpha(K)}}, {\bnle_{\alpha(K)}})$.

\begin{definition}
   \label{def_matchpref_operator}

    Let $\trianglelefteq$ be a total order on $\N \times \N$. Define an
    operator $\phi_{\trianglelefteq}$ according to the choice function
    \begin{equation}
        \label{eqn:match_preference_alpha_definition}
        \alpha_{\trianglelefteq}(K)
        = \argmin_{K' \in \minch{K}}{
            \vect_{\trianglelefteq}(K \oplus K')
        }
    \end{equation}
    where $[K \oplus K']_{ab} = |K_{ab} - K_{ab}'|$, and the minimum is taken
    w.r.t the lexicographic ordering on $\{0,1\}^{|A| \cdot
    |B|}$.\footnotemark{} Operators generated in this way will be called
    \emph{match-preference operators}.

\end{definition}

\footnotetext{
    Note that $K \oplus K'$ is 1 in exactly the entries where $K$ and $K'$
    differ.
}

\begin{example}
    \label{ex:match_preference_example}

    Let $\trianglelefteq$ be the lexicographic order\footnotemark{} on $\N
    \times \N$ so that $\vect_{\trianglelefteq}(K \oplus K')$ is obtained by
    collecting the entries of $K \oplus K'$ row-by-row, from top to
    bottom and left to right.
    Take $K$ from \Cref{ex:minch}. Writing $K_1,\ldots,K_4$ for the elements of
    $\minch{K}$ in the order that they appear in \Cref{ex:minch} and setting
    $v_i = \vect_{\trianglelefteq}(K \oplus K_i)$, we have
    \begin{align*}
        v_1 = (0\bm{{\color{red}1}}00\ 0000\ \bm{{\color{red}1}}0000);
        & \quad \quad
        v_2 = (00\bm{{\color{red}1}}0\ 0000\ \bm{{\color{red}1}}0000) \\
        v_3 = (0000\ 0\bm{{\color{red}1}}00\ \bm{{\color{red}1}}0000);
        & \quad \quad
        v_4 = (0000\ 00\bm{{\color{red}1}}0\ \bm{{\color{red}1}}0000)
    \end{align*}
    The lexicographic minimum is the one with the 1 entries as far right as
    possible, which in this case is $v_4$. Consequently
    $\phi_{\trianglelefteq}$ ranks $K$ according to $K_4$, i.e. $1
    \alt_K^{\phi_{\trianglelefteq}} 2 \alt_K^{\phi_{\trianglelefteq}} 3$ and $1
    \beq_K^{\phi_{\trianglelefteq}} 3 \blt_K^{\phi_{\trianglelefteq}} 2
    \blt_K^{\phi_{\trianglelefteq}} 4$.

\end{example}

\footnotetext{
    That is, $(a, b) \trianglelefteq (a', b')$ iff $a < a'$ or ($a = a'$ and
    $b \le b'$).
}


To conclude the discussion of match-preference operators, we note that one can
compute $\alpha_{\trianglelefteq}(K)$ as the unique closest chain tournament to
$K$ w.r.t a \emph{weighted} Hamming distance, and thereby avoid the need to
enumerate $\minch{K}$ in full as per
\cref{eqn:match_preference_alpha_definition}.

\begin{theorem}
   \label{prop:matchpref_weightings}

    Let $\trianglelefteq$ be a total order on $\N \times \N$. Then for any $m,
    n \in \N$ there exists a function $w: [m] \times [n] \to \R_{\ge 0}$ such
    that for all $m \times n$ tournaments $K$:
    \begin{equation}
        \label{eqn:matchpref_argmin}
        \argmin_{K' \in \ch_{m,n}}{d_w(K, K')} = \{\alpha_{\trianglelefteq}(K)\}
    \end{equation}
    where $d_w(K, K') = \sum_{(a,b) \in [m] \times [n]}{w(a,b) \cdot |K_{ab} -
    K'_{ab}|}$.

\end{theorem}


For example, the weights corresponding to $\trianglelefteq$ from
\Cref{ex:match_preference_example} and $m = 2$, $n = 3$ are
$
    w = \left[\begin{smallmatrix}
        1.5 & 1.25 & 1.125 \\
        1.0625 & 1.03125 & 1.015625
    \end{smallmatrix}\right]
$.

\section{Relaxing chain-min}
\label{sec:relaxing_chain_min}

Having studied chain-minimal operators in some detail, we turn to two remaining
problems: \axiomref{chain-min} is incompatible with \axiomref{anon}, and
computing a chain-minimal operator is \complexityclass{NP}-hard. In this
section we obtain both anonymity and tractability by relaxing the
\axiomref{chain-min} requirement to a property we call
\emph{chain-definability}. We go on to characterise the class of operators with
this weaker property via a greedy approximation algorithm, single out a
particularly intuitive instance, and revisit the axioms of
\Cref{sec:axiomatic_analysis}.

\subsection{Chain-definability}

The source of the difficulties with \axiomref{chain-min} lies in the
minimisation aspect of chain editing. A natural way to retain the spirit of
\axiomref{chain-min} without the complications is to require that $\phi(K)$
corresponds to \emph{some} chain tournament, not necessarily one closest to
$K$. We call this property \emph{chain-definability}.

\begin{axiom}[chain-def]
    For every $m \times n$ tournament $K$ there is $K' \in \ch_{m,n}$ such that
    $\phi(K) = ({\anle_{K'}}, {\bnle_{K'}})$.
\end{axiom}

Clearly \axiomref{chain-min} implies \axiomref{chain-def}. `Chain-definable'
operators can also be cast in the MLE framework of \Cref{sec:mle} as those
whose rankings correspond to \emph{some} (not necessarily MLE) state $\theta$.

At first glance it may seem difficult to determine whether a given pair of
rankings correspond to a chain tournament, since the number of such tournaments
grows rapidly with $m$ and $n$.
Fortunately, \axiomref{chain-def} can be characterised without reference to
chain tournaments by considering the number of \emph{ranks} of ${\ale_K^\phi}$
and ${\ble_K^\phi}$. In what follows $\ranks{\preceq}$ denotes the number of
ranks of a total preorder $\preceq$, i.e. the number of equivalence classes of
its symmetric part.

\begin{theorem}
    \label{result:chain_def_ranks_characterisation}

    $\phi$ satisfies \axiomref{chain-def} if and only if $|\ranks{\ale_K^\phi}
    - \ranks{\ble_K^\phi}| \le 1$ for every tournament $K$.
\end{theorem}

\subsection{Interleaving Operators}
\label{sec:interleaving}

According to \Cref{result:chain_def_ranks_characterisation}, to construct a
chain-definable operator it is enough to ensure that the number of ranks of
$\ale_K^\phi$ and $\ble_K^\phi$ differ by at most one. A simple way to achieve
this is to iteratively select and remove the top-ranked players of $A$ and $B$
simultaneously, until one of $A$ or $B$ is exhausted. We call such operators
\emph{interleaving operators}. Closely related ranking methods have been
previously introduced for non-bipartite tournaments by
\citet{bouyssou2004monotonicity}.

Formally, our procedure is defined by two functions $f$ and $g$ which select
the next top ranks given a tournament $K$ and subsets $A' \subseteq A$, $B'
\subseteq B$ of the remaining players.

\begin{definition}
    \label{def:selectionfunction}

     An $\asymb$-\emph{selection function} is a mapping $f: \K \times 2^\N \times
     2^\N \to 2^\N$ such that for any tournament $K$, $A' \subseteq A$ and $B'
     \subseteq B$:
    \begin{inlinelist}
        \item \label{item:f_sel_1} $f(K, A', B') \subseteq A'$;
        \item \label{item:f_sel_2} If $A' \ne \emptyset$ then $f(K, A', B') \ne
              \emptyset$;
        \item \label{item:f_sel_3} $f(K, A', \emptyset) = A'$
    \end{inlinelist}.

    Similarly, a $\bsymb$-\emph{selection function} is a mapping $g: \K \times 2^\N
    \times 2^\N \to 2^\N$ such that
    \begin{inlinelist}
        \item \label{item:g_sel_1} $g(K, A', B') \subseteq B'$;
        \item \label{item:g_sel_2} If $B' \ne \emptyset$ then $g(K, A', B') \ne
              \emptyset$;
        \item \label{item:g_sel_3} $g(K, \emptyset, B') = B'$
    \end{inlinelist}.

\end{definition}

The corresponding interleaving operator ranks players according to how soon
they are selected in this way; the earlier the better.

\begin{definition}
    \label{def:interleaving}

    Let $f$ and $g$ be selection functions and $K$ a tournament. Write $A_0
    = A$, $B_0 = B$, and for $i \ge 0$:
    \[
        A_{i+1} = A_i \setminus f(K, A_i, B_i);
        \quad
        B_{i+1} = B_i \setminus g(K, A_i, B_i)
    \]
    For $a \in A$ and $b \in B$, write $r(a) = \max{\{ i \mid a \in A_i \}}$
    and $s(b) = \max{\{ i \mid b \in B_i \}}$.\footnotemark{} We define the
    corresponding \emph{interleaving operator} $\phi = \intop{f,g}$ by $a
    \ale_K^\phi a'$ iff $r(a) \ge r(a')$ and $b \ble_K^\phi b'$ iff $s(b) \ge
    s(b')$.

    \footnotetext{
        We show in the appendix that the recursive procedure eventually
        terminates with $A_i$ and $B_i$ becoming empty (and remaining so) after
        finitely many iterations, so $r$ and $s$ are well-defined.
    }

\end{definition}

Note that $A_i$ and $B_i$ are the players left remaining after $i$ applications
of $f$ and $g$, i.e. after removing the top $i$ ranks from both sides. Before
giving a concrete example, we note that interleaving is not just \emph{one} way
to satisfying \axiomref{chain-def}, it is the \emph{only} way.

\begin{theorem}
   \label{result:chaindef_iff_interleaving}

    An operator $\phi$ satisfies \axiomref{chain-def} if and only if $\phi =
    \intop{f,g}$ for some selection functions $(f, g)$.

\end{theorem}

\Cref{result:chaindef_iff_interleaving} justifies our study of interleaving
operators, and provides a different perspective on chain-definability via the
selection functions $f$ and $g$. We come to an important example.

\begin{example}
    \label{ex:cardint}

    Define the \emph{cardinality-based interleaving operator} $\phicardint =
    \intop{f,g}$ where $f(K, A', B') = \argmax_{a \in A'}{|K(a) \cap B'|}$ and
    \\  
    $g(K, A', B') = \argmin_{b \in B'}{|K^{-1}(b) \cap A'|}$, so that the
    `winners' at each iteration are the $A$s with the most wins, and the $B$s
    with the least losses, when restricting to $A'$ and $B'$ only.  We take the
    $\argmin$/$\argmax$ to be the emptyset whenever $A'$ or $B'$ is empty.

    \Cref{tab:cardint_example} shows the iteration of the algorithm for a $4
    \times 5$ tournament $K$. In each row $i$ we show $K$ with the rows and
    columns of $A \setminus A_i$ and $B \setminus B_i$ greyed out, so as to
    make it more clear how the $f$ and $g$ values are
    calculated.\footnotemark{} For brevity we also write $f$ and $g$ in place
    of $f(K, A_i, B_i)$ and $g(K, A_i, B_i)$ respectively.

    The $r$ and $s$ values can be read off as 0, 2, 1, 3 for $A$ and 0, 3, 1,
    1, 2 for $B$, giving the ranking on $A$ as $4 \alt 2 \alt 3 \alt 1$, and
    the ranking on $B$ as $2 \blt 5 \blt 3 \beq 4 \blt 1$. Note also that each
    $f(K, A_i, B_i)$ is a rank of $\ale_K^\phi$ (and similar for $g(K, A_i,
    B_i)$), so the rankings can in fact be read off by looking at the $f$ and
    $g$ columns of \Cref{tab:cardint_example}.

    \footnotetext{
        Note that while $f$ and $g$ for $\phicardint$ are independent of the
		greyed out entries, we do not require this property for selection
        functions in general.
    }

\end{example}

\begin{table}
	\caption{Iteration of the interleaving algorithm for $\phicardint$}
    \label{tab:cardint_example}
    \footnotesize

	\def\g#1{{\color{lightgray}#1}}
	\def\r#1{\bm{{\color{red}#1}}}
	\def\kzero{
		\left[
			\begin{smallmatrix}
				1 & 1 & 1 & 1 & 0 \\
				0 & 1 & 0 & 0 & 1 \\
				0 & 1 & 0 & 1 & 1 \\
				0 & 1 & 1 & 0 & 0 \\
			\end{smallmatrix}
		\right]
	}
	\def\kone{
		\left[
			\begin{smallmatrix}
				\g{1} & \g{1} & \g{1} & \g{1} & \g{0} \\
				\g{0} & 1     & 0     & 0     & 1 \\
				\g{0} & 1     & 0     & 1     & 1 \\
				\g{0} & 1     & 1     & 0     & 0 \\
			\end{smallmatrix}
		\right]
	}
	\def\ktwo{
		\left[
			\begin{smallmatrix}
				\g{1} & \g{1} & \g{1} & \g{1} & \g{0} \\
				\g{0} & 1     & \g{0} & \g{0} & 1 \\
				\g{0} & \g{1} & \g{0} & \g{1} & \g{1} \\
				\g{0} & 1     & \g{1} & \g{0} & 0 \\
			\end{smallmatrix}
		\right]
	}
	\def\kthree{
		\left[
			\begin{smallmatrix}
				\g{1} & \g{1} & \g{1} & \g{1} & \g{0} \\
				\g{0} & \g{1} & \g{0} & \g{0} & \g{1} \\
				\g{0} & \g{1} & \g{0} & \g{1} & \g{1} \\
				\g{0} & 1     & \g{1} & \g{0} & \g{0} \\
			\end{smallmatrix}
		\right]
	}

	\def\kpzero{
		\left[
			\begin{smallmatrix}
				1 & 1 & 1 & 1 & \r{1} \\
				0 & 1 & 0 & 0 & 1     \\
				0 & 1 & 0 & 1 & 1     \\
				0 & 1 & 1 & 0 & 0     \\
			\end{smallmatrix}
		\right]
	}
	\def\kpone{
		\left[
			\begin{smallmatrix}
				1 & 1 & 1     & 1 & \r{1} \\
				0 & 1 & 0     & 0 & 1     \\
				0 & 1 & \r{1} & 1 & 1     \\
				0 & 1 & \r{0} & 0 & 0     \\
			\end{smallmatrix}
		\right]
	}

	\def\es{\emptyset}

	\begin{tabular}{ccccccc}
\toprule
$i$ & $K$        & $A_i$         & $B_i$           & $f$         & $g$       & $K'_i$    \\
\midrule
$0$ & $\kzero$   & $\{1,2,3,4\}$ & $\{1,2,3,4,5\}$ & $\{1\}$     & $\{1\}$   & $\kpzero$ \\[2mm]
$1$ & $\kone$    & $\{2,3,4\}$   & $\{2,3,4,5\}$   & $\{3\}$     & $\{3,4\}$ & $\kpone$  \\[2mm]
$2$ & $\ktwo$    & $\{2,4\}$     & $\{2,5\}$       & $\{2\}$     & $\{5\}$   & -         \\[2mm]
$3$ & $\kthree$  & $\{4\}$       & $\{2\}$         & $\{4\}$     & $\{2\}$   & -         \\[2mm]
$4$ & -          & $\es$         & $\es$    	   & $\es$       & $\es$     & -         \\
\bottomrule
	\end{tabular}
\end{table}

The interleaving algorithm can also be seen as a greedy algorithm for
converting $K$ into a chain graph directly. Indeed, by setting the
neighbourhood of each $a \in f(K, A_i, B_i)$ to $B_i$, and removing each $b \in
g(K, A_i, B_i)$ from the neighbourhoods of all $a \in A_{i+1}$, we eventually
obtain a chain graph. We show this process in the $K'_i$ column of
\Cref{tab:cardint_example}, where only three entries need to be
changed.\footnotemark{} The selection functions $f$ and $g$ can therefore be
seen as \emph{heuristics} with the goal of finding a chain graph `close' to
$K$.

\footnotetext{
    In this example $\minch{K}$ contains a single tournament a distance of 2
    from $K$, so $\phicardint$ makes one more change than necessary.
}

The operator $\phicardint$ from \Cref{ex:cardint} uses simple cardinality-based
heuristics, and can be seen as a chain-definable version of $\phicount$ (which
is not chain-definable). It is also the bipartite counterpart to repeated
applications of Copeland's rule \cite{bouyssou2004monotonicity}. Note that
$f(K, A_i, B_i)$ and $g(K, A_i, B_i)$ can be computed in $O(N^2)$ time at each
iteration $i$, where $N = |A| + |B|$. Since there cannot be more than $N$
iterations, it follows that the rankings of $\phicardint$ can be computed in
$O(N^3)$ time.

\subsection{Axiom Compatibility}

We now revisit the axioms of \Cref{sec:axiomatic_analysis} in relation to
chain-definable operators in general and $\phicardint$ specifically. Firstly,
the weakening of \axiomref{chain-min} pays off: \axiomref{chain-def} is
compatible with all our axioms.

\begin{theorem}
    \label{result:chaindef_axiom_compatibilities}

    For each of \axiomref{anon}, \axiomref{dual}, \axiomref{IIM},
    \axiomref{mon} and \axiomref{pos-resp}, there exists an operator satisfying
    \axiomref{chain-def} and the stated property.

\end{theorem}

\balance
Unfortunately, these cannot all hold at the same time. Indeed, taking
$
    K = \left[\begin{smallmatrix}
        0 & 0 & 1 & 1 \\
        0 & 1 & 0 & 1
    \end{smallmatrix}\right]^\tr
$
and assuming \axiomref{anon} and \axiomref{pos-resp}, the ranking on $A$ is
fully determined as $1 \alt 2 \aeq 3 \alt 4$, and $\ranks{\ale_K^\phi} = 3$.
However, \axiomref{anon} with \axiomref{dual} implies the ranking of $B$ is
flat, i.e.  $\ranks{\ble_K^\phi} = 1$. This contradicts \axiomref{chain-def} by
\Cref{result:chain_def_ranks_characterisation}, yielding the following
impossibility result.

\begin{theorem}
    \label{result:chaindef_impossibility}

    There is no operator satisfying \axiomref{chain-def}, \axiomref{anon},
    \axiomref{dual} and \axiomref{pos-resp}.

\end{theorem}

For the specific operator $\phicardint$ we have the following.

\begin{theorem}
    \label{result:phicardint_axioms}

    $\phicardint$ satisfies \axiomref{chain-def}, \axiomref{anon},
    \axiomref{dual} and \axiomref{mon}, and does not satisfy \axiomref{IIM} or
    \axiomref{pos-resp}.

\end{theorem}

Note that \axiomref{anon} \emph{is} satisfied. This makes $\phicardint$ an
important example of a well-motivated, tractable, chain-definable
and anonymous operator, meeting the criteria outlined at the start of
this section.

\section{Related Work}
\label{sec:related_work}

\inlineheading{On chain graphs}
Chain graphs were originally introduced by \citet{yannakakis1981computing}, who
proved that \emph{chain completion} -- finding the minimum number of edges that
when added to a bipartite graph form a chain graph -- is
\complexityclass{NP}-complete. Hardness results have subsequently been obtained
for chain \emph{deletion}~\cite{natanzon2001complexity} (where only edge
deletions are allowed) and chain \emph{editing}~\cite{drange2015threshold}
(where both additions and deletions are allowed). We refer the reader to the
work of \citet{jiao2017algorithms} and \citet{drange2015threshold} for a more
detailed account of this literature.
Outside of complexity theory, chain graphs have been studied for their spectral
properties in \cite{andelic_2015,ghorbani2017spectral}, and the more general
notion of a \emph{nested colouring} was introduced in \cite{cook2015nested}.

\inlineheading{On tournaments in social choice}
Tournaments have important applications in the design of voting rules, where an
alternative $x$ beats $y$ in a pairwise comparison if a majority of voters
prefer $x$ to $y$.  Various \emph{tournament solutions} have been proposed,
which select a set of `winners' from a given tournament.\footnote{Note that a
ranking, such as we consider in this paper, induces a set of winners by taking
the maximally ranked players.}
Of particular relevance to our work are the \emph{Slater set} and
\emph{Kemeney's rule} \cite{brandt2016a}, which find minimal sets of edges to
invert in the tournament graph such that the beating relation becomes a total
order.\footnotemark{}
These methods are intuitively similar to chain editing: both
involve making minimal changes to the tournament until some property is
satisfied. A rough analogue to the Slater set in our framework is the union of
the top-ranked players from each $K' \in \minch{K}$. Solutions based on the
covering relation -- such as the \emph{uncovered} and \emph{Banks} set
\cite{brandt2016a} -- also bear similarity to chain editing.

\footnotetext{
    Note that like chain editing, Kemeny's rule also admits a maximum
    likelihood characterisation~\cite{elkind2016rationalizations}.
}

Finally, note that directed versions of chain graphs (obtained by orienting
edges from $A$ to $B$ and adding missing edges from $B$ to $A$) correspond to
\emph{acyclic tournaments}, and a topological sort of $A$ becomes a
linearisation of the chain ranking $\anle_K$. This suggests a connection
between chain deletion and the standard \emph{feedback arc set} problem for
removing cycles and obtaining a ranking.

\inlineheading{On generalised tournaments}
A \emph{generalised} tournament~\cite{gonzalez2014paired} is a pair $(X, T)$,
where $X = [t]$ for some $t \in \N$ and $T \in \R_{\ge 0}^{t \times t}$ is a
non-negative $t \times t$ matrix with $T_{ii} = 0$ for all $i \in X$. In this
formalism each encounter between a pair of players $i$ and $j$ is represented
by \emph{two} numbers: $T_{ij}$ and $T_{ji}$. This allows one to model both
intensities of victories and losses (including draws) via the difference
$T_{ij} - T_{ji}$, and the case where a comparison is not available (where
$T_{ij} = T_{ji} = 0$).

Any $m \times n$ bipartite tournament $K$ has a natural generalised tournament
representation via the $(m + n) \times (m + n)$ \emph{anti-diagonal block
matrix}
$
    T = \left[\begin{smallmatrix}
        0 & K \\
        \dual{K} & 0
    \end{smallmatrix}\right]
$, where the top-left and bottom-right blocks are the $m \times m$ and $n
\times n$ zero matrices respectively.  However, such anti-diagonal block
matrices are often excluded in the generalised tournament literature due to an
assumption of \emph{irreducibility}, which requires that the directed graph
corresponding to $T$ is strongly connected. This is not the case in general for
$T$ constructed as above, which means not all existing tournament
operators (and tournament axioms) are well-defined for bipartite
inputs.\footnotemark{} Consequently, bipartite tournaments are a special case
of generalised tournaments \emph{in principle}, but not in practise.

\footnotetext{
    We note that \citet{slutzki2005ranking} side-step the reducibility issue by
    decomposing $T$ into irreducible components and ranking each separately,
    although their methods may give only \emph{partial} orders.
}

\section{Conclusion}
\label{sec:conclusion}

\inlineheading{Summary}
In this paper we studied chain editing, an interesting problem from
computational complexity theory, as a ranking mechanism for bipartite
tournaments. We analysed such mechanisms from a probabilistic viewpoint via the
MLE characterisation, and in axiomatic terms. To resolve both the failure of an
important anonymity axiom and \complexityclass{NP}-hardness, we weakened the
chain editing requirement to one of \emph{chain definability}, and
characterised the resulting class of operators by the intuitive interleaving
algorithm.

\inlineheading{Limitations and future work}
The hardness of chain editing remains a limitation of our approach. A possible
remedy is to look to one of the numerous variant problems that are
polynomial-time solvable~\cite{jiao2017algorithms}; determining their
applicability to ranking is an interesting topic for future work. One could
develop approximation algorithms for chain editing, possibly based on existing
approximations of chain completion \cite{natanzon2000polynomial}. The
interleaving operators of \Cref{sec:interleaving} go in this direction, but we
did not yet obtain any theoretical or experimental bounds on the approximation
ratio.

A second limitation of our work lies in the assumptions of the probabilistic
model; namely that the true state of the world can be reduced to vectors of
numerical skill levels which totally describe the tournament participants. This
assumption may be violated when the competitive element of a tournament is
\emph{multi-faceted}, since a single number cannot represent multiple
orthogonal components of a player's capabilities. Nevertheless, if skill levels
are taken as \emph{aggregations} of these components, chain editing may prove
to be a useful, albeit simplified, model.

Finally, there is room for more detailed axiomatic investigation. In this paper
we have stuck with fairly standard social choice axioms and performed
preliminary analysis. However, the indirect nature of the comparisons in a
bipartite tournament presents unique challenges; new axioms may need to be
formulated to properly evaluate bipartite ranking methods in a normative
sense.

\ifdefined\thisistheprerint
\begin{acks}
We thank the anonymous AAMAS reviewers for their helpful comments.
\end{acks}
\else
\fi






\bibliographystyle{ACM-Reference-Format}
\bibliography{references}


\appendix
\section{Proofs}

This appendix contains proofs that were omitted (or only sketched) in the main
paper.

\subsection{Proof of \Cref{result:mle_hamming}}

The proof of \Cref{result:mle_hamming} requires a lemma of its own.

\begin{lemma}
   \label{result:probexpression}

   Let $K$ be an $m \times n$ tournament, $\bm{\alpha} \in [0,1]^2$ and
   $\theta \in \Theta_{m,n}$. Then
   \begin{equation*}
      \begin{split}
          P_{\bm{\alpha}}(K \mid \theta)
          =
          \prod_{a \in A}
            &\alpha_+^{|K(a) \setminus K_\theta(a)|}
            (1 - \alpha_-)^{|K(a) \cap K_\theta(a)|}
            \\
            &\quad (1 - \alpha_+)^{|B \setminus (K(a) \cup K_\theta(a))|}
            \alpha_-^{|K_\theta(a) \setminus K(a)|}
      \end{split}
   \end{equation*}
\end{lemma}

\begin{proof}
    Write $p_{ab,K}$ for $P_{\bm{\alpha}}(X_{ab} = K_{ab} \mid \theta)$.
    Expanding the product in \cref{def:probdist}, we have
    \[
       P_{\bm{\alpha}}(K \mid \theta)
       = \prod_{a \in A}{\prod_{b \in B}{p_{ab,K}}}
    \]

    Let $a \in A$. Note that $B$ can be written as the disjoint
    union $B = B_1 \cup B_2 \cup B_3 \cup B_4$, where
    \[
       \begin{aligned}
          B_1 &= K(a) \setminus K_\theta(a) \\
          B_2 &= K(a) \cap K_\theta(a) \\
          B_3 &= B \setminus (K(a) \cup K_\theta(a)) \\
          B_4 &= K_\theta(a) \setminus K(a)
       \end{aligned}
    \]
    Recall that $b \in K_\theta(a)$ iff $x_a \ge y_b$
    (where $\theta = \tuple{\bm{x}, \bm{y}}$).  It follows that

    \begin{itemize}
        \item $b \in B_1$ iff $K_{ab} = 1$ and $x_a < y_b$
        \item $b \in B_2$ iff $K_{ab} = 1$ and $x_a \ge y_b$
        \item $b \in B_3$ iff $K_{ab} = 0$ and $x_a < y_b$
        \item $b \in B_4$ iff $K_{ab} = 0$ and $x_a \ge y_b$
    \end{itemize}

    Note that this correspond exactly to the four cases in
    \labelcref{eqn:probdist_random_var_one} and
    \labelcref{eqn:probdist_random_var_zero} which define $p_{ab, K}$; we have
    \[
        p_{ab,K} = \begin{cases}
            \alpha_+,& b \in B_1 \\
            1 - \alpha_-,& b \in B_2 \\
            1 - \alpha_+,& b \in B_3 \\
            \alpha_-,& b \in B_4
        \end{cases}
    \]
    Consequently
    \[
       \begin{aligned}
           \prod_{b \in B}{p_{ab,K}}
           &=
               \left(\prod_{b \in B_1}{\alpha_+}\right)
               \left(\prod_{b \in B_2}{(1-\alpha_-)}\right)
               \left(\prod_{b \in B_3}{(1-\alpha_+)}\right)
               \left(\prod_{b \in B_4}{\alpha_-}\right)
           \\
           &= \alpha_+^{|B_1|} (1-\alpha_-)^{|B_2|} (1-\alpha_+)^{|B_3|}
              \alpha_-^{|B_4|} \\
           &= \alpha_+^{|K(a) \setminus K_\theta(a)|}
              (1-\alpha_-)^{|K(a) \cap K_\theta(a)|}
              \\
           &\quad \quad (1-\alpha_+)^{|B \setminus (K(a) \cup K_\theta(a))|}
              \alpha_-^{|K_\theta(a) \setminus K(a)|}
       \end{aligned}
    \]
    Taking the product over all $a \in A$ we reach the desired
    expression for $P_{\bm{\alpha}}(K \mid \theta)$.
\end{proof}

\begin{proof}[Proof of \Cref{result:mle_hamming}]
    Let $\theta \in \Theta_{m,n}$. From \cref{result:probexpression} we get
    \[
       \begin{aligned}
       P_{\bm{\alpha}}(K \mid \theta)
       = \prod_{a \in A}
            &\beta^{
              |K(a) \setminus K_\theta(a)| + |K_\theta(a) \setminus K(a)|
            }
            \\
            &\quad (1 - \beta)^{
              |K(a) \cap K_\theta(a)|
              + |B \setminus (K(a) \cup K_\theta(a))|
            }
       \end{aligned}
    \]
    Note that
    \[
       |K(a) \setminus K_\theta(a)| + |K_\theta(a) \setminus K(a)|
       =
       |K(a) \symdiff K_\theta(a)|
    \]
    \[
       |K(a) \cap K_\theta(a)|
       + |B \setminus (K(a) \cup K_\theta(a))|
       =
       |B| - |K(a) \symdiff K_\theta(a)|
    \]
    and so
    \[
       \begin{aligned}
           P_{\bm{\alpha}}(K \mid \theta)
           &= \prod_{a \in A}{
                \beta^{
                  |K(a) \symdiff K_\theta(a)|
                }
                (1 - \beta)^{
                  |B| - |K(a) \symdiff K_\theta(a)|
                }
           } \\
           &= \prod_{a \in A}{
                \left(
                  \frac{\beta}{1 - \beta}
                \right)^{
                  |K(a) \symdiff K_\theta(a)|
                }
                (1 - \beta)^{|B|}
           } \\
           &= \underbrace{
                (1 - \beta)^{|A| \cdot |B|}
              }_{=c}
               \prod_{a \in A}{
                \left(
                  \frac{\beta}{1 - \beta}
                \right)^{
                  |K(a) \symdiff K_\theta(a)|
                }
           } \\
           &= c
              \prod_{a \in A}{
               \left(
                 \frac{\beta}{1 - \beta}
               \right)^{
                 |K(a) \symdiff K_\theta(a)|
               }
           }
       \end{aligned}
    \]
    where $c$ is a positive constant that does not depend on $\theta$. Now,
    $P_{\bm{\alpha}}(K \mid \theta)$ is positive, and is maximal when its
    logarithm is. We have
    \[
       \begin{aligned}
           \log{P_{\bm{\alpha}}(K \mid \theta)}
           &=
              \log{c}
              + \sum_{a \in A}{
                  |K(a) \symdiff K_\theta(a)|
                  \log{\left(\frac{\beta}{1 - \beta}\right)}
              } \\
           &=
              \log{c}
              +
              \log{\left(\frac{\beta}{1 - \beta}\right)}
              \sum_{a \in A}{
                  |K(a) \symdiff K_\theta(a)|
              } \\
           &=
              \log{c}
              +
              \log{\left(\frac{\beta}{1 - \beta}\right)}
              d(K, K_\theta)
       \end{aligned}
    \]
    Noting that $\beta < 1/2$ implies $\log{\left(\frac{\beta}{1 -
    \beta}\right)} < 0$, it follows that for any $\theta, \theta' \in
    \Theta_{m,n}$:
    \[
       \begin{aligned}
          P_{\bm{\alpha}}(K \mid \theta)
              \ge P_{\bm{\alpha}}(K \mid \theta')
          &\iff \log{P_{\bm{\alpha}}(K \mid \theta)}
              - \log{P_{\bm{\alpha}}(K \mid \theta')} \ge 0 \\
          &\iff \underbrace{\log{\left(\frac{\beta}{1 - \beta}\right)}}_{<0}
              \left[
                d(K, K_{\theta}) - d(K, K_{\theta'})
              \right]
              \ge 0 \\
          &\iff d(K, K_{\theta}) \le d(K, K_{\theta'})
       \end{aligned}
    \]
    which proves the result.
\end{proof}

\subsection{Proof of \Cref{result:chain_iff_ktheta}}

We need a preliminary result.

\begin{lemma}
   \label{result:ktheta_ordering}

   Let $\theta = \tuple{\bm{x}, \bm{y}} \in \Theta_{m,n}$. Then for all
   $a, a' \in A$ and $b, b' \in B$:

   \begin{enumerate}
       \item $K_\theta(a) \subseteq K_\theta(a')$ iff $x_a \le
             x_{a'}$
       \item $K_\theta^{-1}(b) \supseteq K_\theta^{-1}(b')$ iff $y_b
             \le y_{b'}$.
   \end{enumerate}
\end{lemma}

\begin{proof}

    We prove (1); (2) is shown similarly. Let $a, a' \in A$. First suppose $x_a
    \le x_{a'}$. Let $b \in K_\theta(a)$. Then $y_b \le x_a \le x_{a'}$, so $b
    \in K_\theta(a')$ also. This shows $K_\theta(a) \subseteq K_\theta(a')$.

    Now suppose $K_\theta(a) \subseteq K_\theta(a')$. For the sake of
    contradiction, suppose $x_a > x_{a'}$. By \labelcref{eqn:state_condition_a}
    in the definition of a state (\cref{def:stateworld}), there is $b \in B$
    such that $x_{a'} < y_b \le x_{a}$. But this means $b \in K_\theta(a)
    \setminus K_\theta(a')$, which contradicts $K_\theta(a) \subseteq
    K_\theta(a')$. Thus (1) is proved.
\end{proof}

\begin{proof}[Proof of \Cref{result:chain_iff_ktheta}]
    The ``if'' direction follows from \cref{result:ktheta_ordering} part (1):
    if $\theta = \tuple{\bm{x}, \bm{y}}$ and $a, a' \in A$ then
    either $x_a \le x_{a'}$ -- in which case $K_\theta(a)
    \subseteq K_\theta(a')$ -- or $x_{a'} < x_a$ -- in
    which case $K_\theta(a') \subseteq K_\theta(a)$. Therefore $K_\theta$ has
    the chain property.

    For the ``only if'' direction, suppose $K$ has the chain property.
    Define $\theta = \tuple{\bm{x}, \bm{y}}$ by

    \[
       \begin{aligned}
           x_a &= |\{a' \in A \mid K(a') \subseteq K(a)\}| \\
           y_b &= \begin{cases}
              \min\{x_a \mid a \in K^{-1}(b)\}
                  ,& K^{-1}(b) \ne \emptyset \\
              1 + |A|,& K^{-1}(b) = \emptyset
           \end{cases}
       \end{aligned}
    \]

    It is easily that since the neighbourhood-subset relation ${\anle_K}$ is a
    total preorder, we have $K(a) \subseteq K(a')$ if and only if $x_a \le
    x_{a'}$.  First we show that $K_\theta = K$ by showing that $K_{ab} = 1$ if
    and only if $[K_\theta]_{ab} = 1$. Suppose $K_{ab} = 1$. Then $a \in
    K^{-1}(b)$, so $y_b = \min\{x_{a'} \mid a' \in K^{-1}(b)\} \le x_a$ and
    consequently $[K_\theta]_{ab} = 1$.

    Now suppose $[K_\theta]_{ab} = 1$. Then $x_a \ge y_b$.  We must have
    $K^{-1}(b) \ne \emptyset$; otherwise $y_b = 1 + |A|
    > |A| \ge x_a$. We can therefore take $\hat{a} \in \argmin_{a'
    \in K^{-1}(b)}{x_{a'}}$. By definition of $y_b$, $x_{\hat{a}} = y_b \le
    x_a$. But $x_{\hat{a}} \le x_a$ implies $K(\hat{a}) \subseteq K(a)$; since
    $\hat{a} \in K^{-1}(b)$ this gives $b \in K(\hat{a})$ and $b \in K(a)$,
    i.e. $K_{ab} = 1$. This completes the claim that $K = K_\theta$.

    It only remains to show that $\theta$ satisfies conditions
    \labelcref{eqn:state_condition_a} and \labelcref{eqn:state_condition_b} of
    \cref{def:stateworld}. For \labelcref{eqn:state_condition_a}, suppose $x_a
    < x_{a'}$. Then $K(a) \subset K(a')$, i.e there is $b \in K(a') \setminus
    K(a) = K_\theta(a') \setminus K_\theta(a)$. But $b \in K_\theta(a')$ gives
    $y_b \le x_{a'}$, and $b \not\in K_\theta(a)$ gives $x_a < y_b$; this shows
    that \labelcref{eqn:state_condition_a} holds.

    For \labelcref{eqn:state_condition_b}, suppose $y_b < y_{b'}$. Clearly
    $K^{-1}(b) \ne \emptyset$ (otherwise $y_b = 1 + |A|$ is maximal). Thus
    there is $a \in K^{-1}(b)$ such that $y_b = x_a$. This of course means $x_a
    < y_{b'}$; in particular we have $y_b \le x_a < y_{b'}$ as required for
    \labelcref{eqn:state_condition_b}.

    We have shown that $K = K_\theta$ and that $\theta \in \Theta_{m,n}$, and the
    proof is complete.
\end{proof}

\subsection{Proof of \Cref{result:mle_iff_chainmin_operator}}

\begin{proof}

    First we show that for any $m, n \in \N$ and any $m \times n$ tournament
    $K$ it holds that $\theta$ is an MLE state for $K$ if and only if $K_\theta
    \in \minch{K}$.

    Indeed, fix some $m, n$ and $K$. Write $\K_{\Theta_{m,n}} = \{K_\theta \mid
    \theta \in \Theta_{m,n}\}$. By \cref{result:mle_hamming}, $\theta$ is an MLE
    if and only if $d(K, K_\theta) \le d(K, K_{\theta'})$ for all $\theta' \in
    \Theta_{m,n}$, i.e. $K_\theta \in \argmin_{K' \in \K_{\Theta_{m,n}}}{d(K,
    K')}$. But by \cref{result:chain_iff_ktheta}, $\K_{\Theta_{m,n}}$ is just
    $\ch_{m,n}$, the set of all $m \times n$ tournaments with the chain
    property. We see that $\argmin_{K' \in \K_{\Theta_{m,n}}}{d(K, K')} =
    \argmin_{K' \in \ch_{m,n}}{d(K, K')} = \minch{K}$ by definition of
    $\minch{K}$. This shows that $\theta$ is an MLE iff $K_\theta \in
    \minch{K}$.

    Now, by definition, $\phi$ satisfies \axiomref{chain-min} iff for every
    tournament $K$ there is $K' \in \minch{K}$ such that $\phi(K) =
    ({\anle_{K'}}, {\bnle_{K'}})$. Using \cref{result:chain_iff_ktheta} and the
    above result, $K' \in \minch{K}$ if and only if $K' = K_\theta$ for some
    MLE $\theta$ for $K$. We see that \axiomref{chain-min} can be equivalently
    stated as follows: for all tournament $K$ there exists an MLE $\theta$ such
    that $\phi(K) = ({\anle_{K_\theta}}, {\bnle_{K_\theta}})$. But by
    \cref{result:ktheta_ordering} we have $a \anle_{K_\theta} a'$ iff $x_a \le
    x_{a'}$ and $b \bnle_{K_\theta} b'$ iff $y_b \le y_{b'}$ (where $\theta =
    \tuple{\bm{x}, \bm{y}}$). The above reformulation of \axiomref{chain-min}
    now coincides with the definition of a maximum likelihood operator, and we
    are done.
\end{proof}

\subsection{Proof of \Cref{result:chainmin_axiom_incompatibilities}}

\begin{proof}

    We take each axiom in turn. Let $\phi$ be any operator satisfying
    \axiomref{chain-min}.

    \axiomref{anon:} Consider $K = \left[\begin{smallmatrix} 1&0\\0&1
    \end{smallmatrix}\right]$, and define permutations $\sigma = \pi = (1\ 2)$,
    i.e. the permutations which simply swap 1 and 2. It is easily seen that
    $\pi(\sigma(K)) = K$. Supposing $\phi$ satisfied \axiomref{anon}, we would
    get $1 \ale_K^\phi 2$ iff $\sigma(1) \ale_{\pi(\sigma(K))}^\phi \sigma(2)$
    iff $2 \ale_K^\phi 1$, which implies $1 \aeq_K^\phi 2$.
    On the other hand, we have
    \[
        \minch{K} = \left\{
           \left[\begin{smallmatrix}
               1 & \color{red}{1} \\
               0 & 1
           \end{smallmatrix}\right],
           \left[\begin{smallmatrix}
               1 & 0 \\
               \color{red}{1} & 1
           \end{smallmatrix}\right],
           \left[\begin{smallmatrix}
               1 & 0 \\
               0 & \color{red}{0}
           \end{smallmatrix}\right],
           \left[\begin{smallmatrix}
               \color{red}{0} & 0 \\
               0 & 1
           \end{smallmatrix}\right]
        \right\}
    \]
    Since $\phi$ satisfies \axiomref{chain-min} and $1, 2 \in A$ rank equally
    in ${\ale_K^\phi}$, there must be $K' \in \minch{K}$ such that 1 and 2 rank
    equally in ${\anle_{K'}}$, i.e. $K'(1) = K'(2)$. But clearly there is no
    such $K'$; all tournaments in $\minch{K}$ have distinct first and second
    rows. Hence $\phi$ cannot satisfy \axiomref{anon}.

    \axiomref{IIM:} Suppose $\phi$ satisfies \axiomref{chain-min} and
    \axiomref{IIM}. Write
    \[
         K_1 = \left[\begin{smallmatrix}
            1 & 0 & 0 \\
            0 & 1 & 0 \\
            0 & 1 & 1
         \end{smallmatrix}\right]
         , \quad
         K_2 = \left[\begin{smallmatrix}
            1 & 0 & 0 \\
            0 & 1 & 0 \\
            1 & 0 & 1
         \end{smallmatrix}\right]
    \]
    Note that the first and second rows of $K_1$ and $K_2$ are identical, so by
    \axiomref{IIM} we have $1 \ale_{K_1}^\phi 2$ iff $1 \ale_{K_2}^\phi 2$.
    Both tournaments have a unique closest chain tournament requiring changes
    to only a single entry:
    \[
        \minch{K_1} = \left\{
            \left[\begin{smallmatrix}
               \color{red}{0} & 0 & 0 \\
               0 & 1 & 0 \\
               0 & 1 & 1
            \end{smallmatrix}\right]
        \right\}
        , \quad
        \minch{K_2} = \left\{
            \left[\begin{smallmatrix}
               1 & 0 & 0 \\
               0 & \color{red}{0} & 0 \\
               1 & 0 & 1
            \end{smallmatrix}\right]
        \right\}
    \]
    Write ${K_1}'$ and ${K_2}'$ for these nearest chain tournaments
    respectively. By \axiomref{chain-min}, we must have $\phi(K_i) =
    ({\anle_{{K_i}'}}, {\bnle_{{K_i}'}})$. In particular, $1
    \alt_{K_1}^\phi 2$ and $2 \alt_{K_2}^\phi 1$. But this contradicts
    \axiomref{IIM}, and we are done.

    \axiomref{pos-resp:} Suppose $\phi$ satisfies \axiomref{chain-min} and
    \axiomref{pos-resp}, and consider
    \[
        K = \left[\begin{smallmatrix}
            1 & 1 & 1 \\
            1 & 1 & 0 \\
            0 & 0 & 1 \\
            0 & 0 & 1
        \end{smallmatrix}\right]
    \]
    $K$ has a unique closest chain tournament $K'$:
    \[
        \minch{K} = \{K'\} = \left\{
            \left[\begin{smallmatrix}
            1 & 1 & 1 \\
            1 & 1 & \color{red}{1} \\
            0 & 0 & 1 \\
            0 & 0 & 1
        \end{smallmatrix}\right]
        \right\}
    \]
    \axiomref{chain-min} therefore implies $\phi(K) = ({\anle_{K'}},
    {\bnle_{K'}})$.  Note that $K'(1) = K'(2)$, so we have $1 \aeq_K^\phi 2$.
    In particular, $1 \ale_K^\phi 2$. Since $K_{23} = 0$, we may apply
    \axiomref{pos-resp} to get $1 \alt_{K + \bm{1}_{23}}^\phi 2$.  But $K +
    \bm{1}_{23}$ is just $K'$. Since the chain property already holds for
    $K'$, we have $\minch{K'} = \{K'\}$ and consequently
    \[
        \phi(K + \bm{1}_{23})
        = \phi(K')
        = ({\anle_{K'}}, {\bnle_{K'}})
        = \phi(K)
    \]
    so in fact $1 \aeq_{K + \bm{1}_{23}}^\phi 2$, contradicting
    \axiomref{pos-resp}.
\end{proof}

\subsection{Proof of \Cref{result:chainmin_axiom_compatibilities}}

For ease of presentation we establish the compatibility of \axiomref{chain-min}
with \axiomref{dual} and \axiomref{mon} separately.

\begin{proposition}
    \label{result:chainmin_dual_compatibility}
    There exists an operator $\phi$ satisfying \axiomref{chain-min} and
    \axiomref{dual}.
\end{proposition}
\begin{proposition}
    \label{result:chainmin_mon_compatibility}
    There exists an operator $\phi$ satisfying \axiomref{chain-min} and
    \axiomref{mon}.
\end{proposition}

It is clear that these two propositions will together prove
\Cref{result:chainmin_axiom_compatibilities}. For
\Cref{result:chainmin_dual_compatibility} we use the following result.

\begin{lemma}
    \label{result:chainmin_dual_lemma}
    Let $K$ be a tournament. Then
    \begin{enumerate}
        \item ${\bnle_K} = {\anle_{\dual{K}}}$ \label{item:dual_lemma_nle}
        \item $K' \in \minch{K}$ if and only if $\dual{K'} \in
              \minch{\dual{K}}$
              \label{item:dual_lemma_minch}
    \end{enumerate}
\end{lemma}

\begin{proof}
    Fix an $m \times n$ tournament $K$.

    (\labelcref{item:dual_lemma_nle})
    Note that for any $b \in B$, we have $K^{-1}(b) = A \setminus \dual{K}(b)$.
    Indeed, for any $a \in A = A_K = B_{\dual{K}}$,
    \begin{align*}
        a \in K^{-1}(b)
        &\iff K_{ab} = 1 \\
        &\iff 1 - K_{ab} = 0 \\
        &\iff \dual{K}_{ba} = 0 \\
        &\iff a \notin \dual{K}(b)
    \end{align*}
    This means that for any $b, b' \in B$,
    \begin{align*}
        b \bnle_K b'
        &\iff K^{-1}(b) \supseteq K^{-1}(b') \\
        &\iff A \setminus \dual{K}(b) \supseteq A \setminus \dual{K}(b') \\
        &\iff \dual{K}(b) \subseteq \dual{K}(b') \\
        &\iff b \anle_{\dual{K}} b'
    \end{align*}
    so ${\bnle_K} = {\anle_{\dual{K}}}$.

    (\labelcref{item:dual_lemma_minch})
    (⇒) Suppose $K' \in \minch{K}$. First we show that $\dual{K'}$ has the
    chain property. It is sufficient to show that ${\bnle_{K'}}$ is a total
    preorder,\footnotemark{} since part (\labelcref{item:dual_lemma_nle}) then
    implies ${\anle_{\dual{K'}}}$ is a total preorder and $\dual{K'}$ has the
    chain property by definition.

    \footnotetext{
        Note that we claim this holds for any $K'$ with the chain property
        in the body of the paper, but this has not yet been proven.
    }

    Since ${\bnle_{K'}}$ always has reflexivity and transitivity, we only need
    to show the totality property. Let $b, b' \in B$ and suppose $b
    \not\bnle_{K'} b'$. We must show $b' \bnle_{K'} b$, i.e. $(K')^{-1}(b')
    \supseteq (K')^{-1}(b)$. To that end, let $a \in (K')^{-1}(b)$.

    Since $(K')^{-1}(b) \not\supseteq (K')^{-1}(b')$, there is some $\hat{a}
    \in (K')^{-1}(b')$ with $\hat{a} \notin (K')^{-1}(b)$. That is, $b' \in
    K'(\hat{a})$ but $b \notin K'(\hat{a})$. Since $b \in K'(a)$, we have
    $K'(a) \not\subseteq K'(\hat{a})$. By the chain property for $K'$, we get
    $K'(\hat{a}) \subset K'(a)$. Finally, this means $b' \in K'(\hat{a})
    \subseteq K'(a)$, i.e $a \in (K')^{-1}(b')$. This shows $b' \bnle_{K'} b$
    as required.

    It remains to show that $d(\dual{K}, \dual{K'})$ is minimal. Since every
    tournament is the dual of its dual, any $n \times m$ chain tournament is of
    the form $\dual{K''}$ for an $m \times n$ tournament $K''$. The above
    argument shows that the chain property is preserved by taking the dual, so
    that $K''$ has the chain property also. Since $K' \in \minch{K}$, we have
    $d(K, K'') \ge d(K, K')$. It is easily verified that the Hamming distance
    is also preserved under duals, so
    \[
        d(\dual{K}, \dual{K'})
        = d(K, K')
        \le d(K, K'')
        = d(\dual{K}, \dual{K''})
    \]
    We have shown that $\dual{K'}$ is as close to $\dual{K}$ as any other $n
    \times m$ tournament with the chain property, which shows $\dual{K'} \in
    \minch{\dual{K}}$ as required.

    (⇐) Suppose $\dual{K'} \in \minch{\dual{K}}$. By the `only if' statement
    above, we have $\dual{\dual{K'}} \in \minch{\dual{\dual{K}}}$. But
    $\dual{\dual{K}} = K$ and $\dual{\dual{K'}} = K'$, so $K' \in \minch{K}$ as
    required.
\end{proof}

\begin{proof}[Proof of \Cref{result:chainmin_dual_compatibility}]
    Let $\phi$ be an arbitrary operator satisfying \axiomref{chain-min}. Then
    there is a function $\alpha: \K \to \K$ such that $\phi(K) =
    ({\anle_{\alpha(K)}}, {\bnle_{\alpha(K)}})$ and $\alpha(K) \in \minch{K}$
    for all tournaments $K$. We will construct a new function $\alpha'$, based
    on $\alpha$, such that $\alpha'(\dual{K}) = \dual{\alpha'(K)}$.

    Let $\ll$ be a total order on the set of all tournaments
    $\K$.\footnotemark{} Write
    \[
        T = \{K \in \K \mid K \ll \dual{K}\}
    \]
    Note that since $K \ne \dual{K}$ for all $K$, exactly one of $K$ and
    $\dual{K}$ lies in $T$. Informally, we view the tournaments in $T$ as
    somehow `canonical', and those in $\K \setminus T$ as the dual of a
    canonical tournament. We use this notion to define $\alpha'$:
    \[
        \alpha'(K) = \begin{cases}
            \alpha(K),& K \in T \\
            \dual{\alpha(\dual{K})},& K \notin T
        \end{cases}
    \]
    First we claim $\alpha'(K) \in \minch{K}$ for all $K$. Indeed, if $K \in T$
    then $\alpha'(K) = \alpha(K) \in \minch{K}$ by the assumption on $\alpha$.
    Otherwise, $\alpha(\dual{K}) \in \minch{\dual{K}}$, so
    \Cref{result:chainmin_dual_lemma} part (\labelcref{item:dual_lemma_minch})
    implies $\alpha'(K) = \dual{\alpha(\dual{K})} \in \minch{\dual{\dual{K}}} =
    \minch{K}$.

    Next we show $\dual{\alpha'(K)} = \alpha'(\dual{K})$. First suppose $K \in
    T$. Then $\alpha'(K) = \alpha(K)$ and $\dual{K} \notin T$, so
    $\alpha'(\dual{K}) = \dual{\alpha(\dual{\dual{K}})} = \dual{\alpha(K)} =
    \dual{\alpha'(K)}$ as required. Similarly, if $K \notin T$ then $\dual{K}
    \in T$, so $\alpha'(\dual{K}) = \alpha(\dual{K})$, and $\alpha'(K) =
    \dual{\alpha(\dual{K})} = \dual{\alpha'(\dual{K})}$. Taking the dual of
    both sides, we get $\dual{\alpha'(K)} = \alpha'(\dual{K})$.

    Finally, define a new operator $\phi'$ by $\phi'(K) =
    ({\anle_{\alpha'(K)}}, {\bnle_{\alpha'(K)}})$. Since $\alpha'(K) \in
    \minch{K}$ for all $K$, $\phi'$ satisfies \axiomref{chain-min}. Moreover,
    using \Cref{result:chainmin_dual_lemma} part
    (\labelcref{item:dual_lemma_nle}) and the fact that $\dual{\alpha'(K)} =
    \alpha'(\dual{K})$, for any tournament $K$ and $b, b' \in B$ we have
    \begin{align*}
        b \ble_K^{\phi'} b'
        &\iff b \bnle_{\alpha'(K)} b' \\
        &\iff b \anle_{\dual{\alpha'(K)}} b' \\
        &\iff b \anle_{\alpha'(\dual{K})} b' \\
        &\iff b \ble_{\dual{K}}^{\phi'} b'
    \end{align*}
    which shows $\phi'$ also satisfies \axiomref{dual}.
    \footnotetext{
        Note that $\K$ is countable, so such an order can be easily
        constructed. Alternatively, one could use the axiom of choice and
        appeal to the well-ordering theorem to obtain $\ll$.
    }
\end{proof}

Next we prove \Cref{result:chainmin_mon_compatibility}. We will proceed in
three stages. First, \Cref{result:chainmin_mon_swapping} shows that if $K(a_1)
\subseteq K(a_2)$ and $K' \in \minch{K}$ is some closest chain tournament with
the reverse inclusion $K'(a_2) \subseteq K'(a_1)$, then swapping $a_1$ and
$a_2$ in $K'$ yields obtain another closest chain tournament $K'' \in
\minch{K}$. Next, we show in \Cref{result:chainmin_mon_extend_strict_part} that
by performing successive swaps in this way, we can find $K' \in \minch{K}$ such
that $K'(a_1) \subseteq K'(a_2)$ whenever $K(a_1) \subset K(a_2)$ (note the
strict inclusion). Finally, we modify this $K'$ in
\Cref{result:chainmin_mon_extend_full} to additionally satisfy $K'(a_1) =
K'(a_2)$ whenever $K(a_1) = K(a_2)$.  This shows that there always exist an
element of $\minch{K}$ extending the neighbourhood-subset relation $\anle_K$,
and consequently it is possible to satisfy \axiomref{chain-min} and
\axiomref{mon} simultaneously.

\begin{definition}
    Let $K$ be a tournament and $a_1, a_2 \in A$. We denote by
    $\swap{K}{a_1}{a_2}$ the tournament obtained by swapping the $a_1$ and
    $a_2$-th rows of $K$, i.e.
    \[
        [\swap{K}{a_1}{a_2}]_{ab} = \begin{cases}
            K_{a_1,b},& a = a_2 \\
            K_{a_2,b},& a = a_1 \\
            K_{a,b},& a \notin \{a_1,a_2\}
        \end{cases}
    \]
\end{definition}

\begin{lemma}
    \label{result:chainmin_mon_swapping}
    Suppose $K(a_1) \subseteq K(a_2)$ and $K' \in \minch{K}$ is such that
    $K'(a_2) \subseteq K'(a_1)$. Then $\swap{K'}{a_1}{a_2} \in \minch{K}$.
\end{lemma}

\def\xset{(-1, 0) circle (1cm)}
\def\yset{(1, 0) circle (1cm)}
\def\xprimeset{(0, 1) circle (1cm)}
\def\yprimeset{(0, -1) circle (1cm)}
\newcommand{\setboundaries}{
    \draw \xset node {$X$};
    \draw \yset node {$Y$};
    \draw \xprimeset node {$X'$};
    \draw \yprimeset node {$Y'$};
}
\newcommand{\thiswithoutthose}[2]{
    \begin{scope}[even odd rule]
        \clip #2 (-2, -2) rectangle (2, 2);
        \fill[orange!60] #1;
    \end{scope}
}

\begin{figure*}
    \centering
    \subfloat[$X \symdiff (X' \cup Y')$]{%
            \begin{tikzpicture}
                \thiswithoutthose{\xset}{\xprimeset \yprimeset}
                \thiswithoutthose{\xprimeset}{\xset}
                \thiswithoutthose{\yprimeset}{\xset}
                \setboundaries
            \end{tikzpicture}
    }
    \quad
    \subfloat[$(X \cup Y) \symdiff X'$]{%
            \begin{tikzpicture}
                \thiswithoutthose{\xset}{\xprimeset}
                \thiswithoutthose{\yset}{\xprimeset}
                \thiswithoutthose{\xprimeset}{\xset \yset}
                \setboundaries
            \end{tikzpicture}
    }
    \quad
    \subfloat[$X \symdiff X'$]{%
            \begin{tikzpicture}
                \thiswithoutthose{\xset}{\xprimeset}
                \thiswithoutthose{\xprimeset}{\xset}
                \setboundaries
            \end{tikzpicture}
    }
    \quad
    \subfloat[$(X \cup Y) \symdiff (X' \cup Y')$]{%
            \begin{tikzpicture}
                \thiswithoutthose{\xset}{\xprimeset \yprimeset}
                \thiswithoutthose{\yset}{\xprimeset \yprimeset}
                \thiswithoutthose{\xprimeset}{\xset \yset}
                \thiswithoutthose{\yprimeset}{\xset \yset}
                \setboundaries
            \end{tikzpicture}
    }
    \Description{Depictions of the sets in \Cref{eqn:swap_lemma_symdiffs}}
    \caption{Depictions of the sets in \Cref{eqn:swap_lemma_symdiffs}}
    \label{fig:swap_lemma_venns}

\end{figure*}
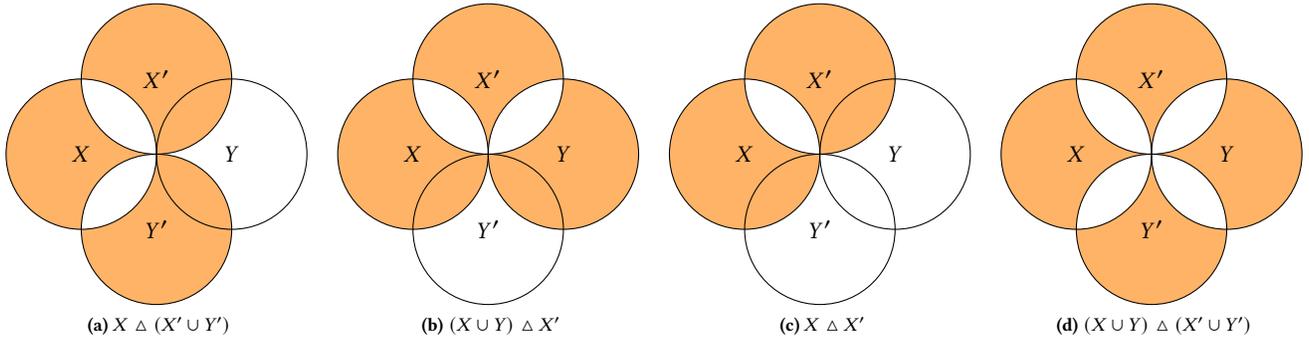

\begin{proof}
    Write $K'' = \swap{K'}{a_1}{a_2}$. It is clear that $K''$ has the chain
    property since $K'$ does. Since $K' \in \minch{K}$, we have $d(K, K'') \ge
    d(K, K')$. We will show that $d(K, K'') \le d(K, K')$ also, which implies
    $d(K, K'') = d(K, K') = \mindist{K}$ and thus $K'' \in \minch{K}$.

    To that end, observe that for any tournament $\hat{K}$,
    \[
        d(K, \hat{K}) = \sum_{a \in A}{|K(a) \symdiff \hat{K}(a)|}
    \]
    Noting that $K'(a) = K''(a)$ for $a \notin \{a_1,a_2\}$ and $K''(a_1) =
    K'(a_2)$, $K''(a_2) = K'(a_1)$, we have
    \begin{align*}
        d(K, K') - d(K, K'')
        &= \sum_{i \in \{1,2\}}{\left(
            |K(a_i) \symdiff K'(a_i)| - |K(a_i) \symdiff K''(a_i)|
        \right)} \\
        &= |K(a_1) \symdiff K'(a_1)| - |K(a_1) \symdiff K'(a_2)| \\
        &\quad + |K(a_2) \symdiff K'(a_2)| - |K(a_2) \symdiff K'(a_1)|
    \end{align*}

    To simplify notation, write $X = K(a_1)$, $X' = K'(a_2)$, $Y = K(a_2)
    \setminus K(a_1)$ and $Y' = K'(a_1) \setminus K'(a_2)$ so that
    \begin{align*}
        K(a_1) = X;
            &\quad\quad
        K(a_2) = X \cup Y \\
        K'(a_1) = X' \cup Y';
            &\quad\quad
        K'(a_2) = X'
    \end{align*}
    and $X \cap Y = X' \cap Y' = \emptyset$. Rewriting the above we have
    \begin{align*}
        d(K, K') - d(K, K'')
        &= |K(a_1) \symdiff K'(a_1)|
           + |K(a_2) \symdiff K'(a_2)|
            \\
        &  \quad
           - |K(a_1) \symdiff K'(a_2)|
           - |K(a_2) \symdiff K'(a_1)| \\
        &= |X \symdiff (X' \cup Y')|
           + |(X \cup Y) \symdiff X'| \\
        &  \quad
           -
           |X \symdiff X'|
           - |(X \cup Y) \symdiff (X' \cup Y')|
           \addtocounter{equation}{1}\tag{\theequation}
           \label{eqn:swap_lemma_symdiffs}
    \end{align*}

    Each of the symmetric differences in \cref{eqn:swap_lemma_symdiffs} are
    depicted in \Cref{fig:swap_lemma_venns}. Note that each of these sets can
    be expressed as a union of the 8 disjoint subsets of $X \cup Y \cup X' \cup
    Y'$ shown in the figure. Expanding the symmetric differences in
    \cref{eqn:swap_lemma_symdiffs} and consulting \Cref{fig:swap_lemma_venns},
    it can be seen that most terms cancel out, and in fact we are left with
    \[
        d(K, K') - d(K, K'') = 2|Y \cap Y'|  \ge 0
    \]
    This shows that $d(K, K'') \le d(K, K')$, and the proof is complete.
\end{proof}

\begin{notation}
    For a relation $R$ on a set $X$ and $x \in X$, write
    \[ U(x, R) = \{y \in X \mid x \mathrel{R} y\} \]
    \[ L(x, R) = \{y \in X \mid y \mathrel{R} x\} \]
    for the upper- and lower-sets of $x$ respectively.
\end{notation}

\begin{lemma}
    \label{result:chainmin_mon_extend_strict_part}
    For any tournament $K$ there is $K' \in \minch{K}$ such that for all $a \in
    A$:
    \[
        U(a, {\anlt_K}) \subseteq U(a, {\anle_{K'}})
    \]
    That is, $K(a) \subset K(a')$ implies $K'(a) \subseteq K'(a')$ for all $a,
    a' \in A$.
\end{lemma}

\begin{proof}
    Write $A = \{a_1,\ldots,a_m\}$, ordered such that $|L(a_1, {\anle_K})| \le
    \cdots \le |L(a_m, {\anle_K})|$. We will show by induction that for each $0
    \le i \le m$ there is $K_i \in \minch{K}$ such that:
    \[
        1 \le j \le i
        \implies
        U(a_j, {\anlt_K}) \subseteq U(a_j, {\anle_{K_i}})
        \tag{$\ast$}
        \label{eqn:mon_lemma_induction}
    \]
    The result follows by taking $K' = K_m$.

    The case $i=0$ is vacuously true, and we may take $K_0$ to be an arbitrary
    member of $\minch{K}$. For the inductive step, suppose
    \labelcref{eqn:mon_lemma_induction} holds for some $0 \le i < m$. If
    $U(a_{i+1}, {\anlt_K}) = \emptyset$ then we may set $K_{i+1} = K_i$, so
    assume that $U(a_{i+1}, {\anlt_K})$ is non-empty. Take some $\hat{a} \in
    \min(U(a_{i+1}, {\anlt_K}), {\anle_{K_i}})$. Then $\hat{a}$ has (one of)
    the smallest neighbourhoods in $K_i$ amongst those in $A$ with a strictly
    larger neighbourhood than $a_{i+1}$ in $K$.

    If $K_i(a_{i+1}) \subseteq K_i(\hat{a})$ then we claim
    \labelcref{eqn:mon_lemma_induction} holds with $K_{i+1} = K_i$. Indeed, for
    $j < i + 1$ the inclusion in \labelcref{eqn:mon_lemma_induction} holds
    since it does for $K_i$. For $j = i+1$, let $a \in U(a_{i+1}, {\anlt_K})$.
    The definition of $\hat{a}$ implies $K_i(a) \not\subset K_i(\hat{a})$;
    since $K_i$ has the chain property this means $K_i(\hat{a}) \subseteq
    K_i(a)$. Consequently $K_i(a_{i+1}) \subseteq K_i(\hat{a}) \subseteq
    K_i(a)$, i.e. $a \in U(a_{i+1}, {\anle_{K_i}}) = U(a_{i+1},
    {\anle_{K_{i+1}}})$ as required.

    For the remainder of the proof we therefore suppose $K_i(a_{i+1})
    \not\subseteq K_i(\hat{a})$. The chain property for $K_i$ gives
    $K_i(\hat{a}) \subset K_i(a_{i+1})$. Since $K_i \in \minch{K}$ and
    $K(a_{i+1}) \subset K(\hat{a})$, we may apply
    \Cref{result:chainmin_mon_swapping}. Set $K_{i+1} =
    \swap{K_i}{a_{i+1}}{\hat{a}} \in \minch{K}$. The inclusion in
    \labelcref{eqn:mon_lemma_induction} is easy to show for $j=i+1$: if $a \in
    U(a_{i+1}, {\anlt_K})$ then either $a = \hat{a}$ -- in which case
    $K_{i+1}(a_{i+1}) \subset K_{i+1}(a)$ by construction -- or $a \ne \hat{a}$
    and $K_{i+1}(a_{i+1}) = K_i(\hat{a}) \subseteq K_i(a) = K_{i+1}(a)$. In
    either case $a \in U(a_{i+1}, {\anle_{K_{i+1}}})$ as required.

    Now suppose $1 \le j < i + 1$. First note that due to our assumption on the
    ordering of $\{a_1,\ldots,a_m\}$, we have $a_j \ne \hat{a}$ (indeed, if
    $a_j = \hat{a}$ then $K(a_{i+1}) \subset K(a_j)$ and $|L(a_j, {\anlt_K})| >
    |L(a_{i+1}, {\anlt_K})|$). Since $a_j \ne a_{i+1}$ also, $a_j$ was not
    involved in the swapping in the construction of $K_{i+1}$, and consequently
    $K_{i+1}(a_j) = K_i(a_j)$. Let $a \in U(a_j, {\anlt_K})$. We must show that
    $K_{i+1}(a_j) \subseteq K_{i+1}(a)$. We consider cases.

    \textbf{Case 1:} $a = \hat{a}$. Using the fact that
    \labelcref{eqn:mon_lemma_induction} holds for $K_i$ we have
    \[
        K_{i+1}(a_j)
        = K_i(a_j)
        \subseteq K_i(\hat{a})
        \subset K_i(a_{i+1})
        = K_{i+1}(\hat{a})
    \]

    \textbf{Case 2:} $a = a_{i+1}$. Here $K(a_j) \subset K(a_{i+1}) \subset
    K(\hat{a})$, i.e. $\hat{a} \in U(a_j, {\anlt_K})$. Applying the inductive
    hypothesis again we have
    \[
        K_{i+1}(a_j)
        = K_i(a_j)
        \subseteq K_i(\hat{a})
        = K_{i+1}(a_{i+1})
    \]

    \textbf{Case 3:} $a \notin \{\hat{a}, a_{i+1}\}$. Here neither $a_j$ nor
    $a$ were involved in the swap, so $K_{i+1}(a_j) = K_i(a_j) \subseteq K_i(a)
    = K_{i+1}(a)$.

    By induction, the proof is complete.
\end{proof}

\begin{lemma}
    \label{result:chainmin_mon_extend_full}
    Let $K$ be a tournament and suppose $K' \in \minch{K}$ is such that $U(a,
    {\anlt_K}) \subseteq U(a, {\anle_{K'}})$ for all $a \in A$. Then there is
    $K'' \in \minch{K}$ such that ${\anle_K} \subseteq {\anle_{K''}}$.


\end{lemma}

\begin{proof}
    Let $A_1, \ldots, A_t \subseteq A$ be the equivalence classes of
    ${\aneq_K}$, the symmetric part of ${\anle_K}$. Note that $a \aneq_K a'$
    iff $K(a) = K(a')$, so we can associate each $A_i$ with a neighbourhood
    $B_i \subseteq B$ such that $K(a) = B_i$ whenever $a \in A_i$.

    Our aim is to select a single element from each equivalence class $A_i$,
    which we denote by $f(A_i)$, and modify $K'$ to set the neighbourhood of
    each $a \in A_i$ to $K'(f(A_i))$. To that end, construct a function $f:
    \{A_1,\ldots,A_t\} \to A$ such that
    \[
        f(A_i) \in \argmin_{a \in A_i}{|B_i \symdiff K'(a)|} \in A_i
    \]
    Define $K''$ by $K''_{ab} = K'_{f([a]), b}$, where $[a]$ denotes the
    equivalence class of $a$. Then $K''(a) = K'(f([a]))$ for all $a$.

    Next we show that $K'' \in \minch{K}$. Note that $K''$ has the chain
    property, since $a_1 \anle_{K''} a_2$ iff $f([a_1]) \anle_{K'} f([a_2])$,
    and $f([a_1]), f([a_2])$ are guaranteed to be comparable with respect to
    ${\anle_{K'}}$ since $K'$ has the chain property. To show $d(K, K'')$ is
    minimal, observe that
    \begin{align*}
        d(K, K'')
        &= \sum_{a \in A}{|K(a) \symdiff K''(a)|} \\
        &= \sum_{i=1}^{t}{
            \sum_{a \in A_i}{
                |B_i \symdiff K'(f(A_i))|
            }
        }
    \end{align*}
    By definition of $f$, we have $|B_i \symdiff K'(f(A_i))| \le |B_i \symdiff
    K'(a)|$ for all $a \in A_i$. Consequently
    \begin{align*}
        d(K, K'')
        &\le \sum_{i=1}^{t}{
            \sum_{a \in A_i}{
                |B_i \symdiff K'(a)|
            }
        } \\
        &= d(K, K') \\
        &= \mindist{K}
    \end{align*}
    which implies $K'' \in \minch{K}$.

    We are now ready to prove the result. Suppose $a \anle_K a'$ i.e. $K(a)
    \subseteq K(a')$. If $K(a) = K(a')$ then $[a] = [a']$, so
    \[
        K''(a) = K'(f([a])) = K'(f([a'])) = K''(a')
    \]
    and in particular $K''(a) \subseteq K''(a')$. If instead $K(a) \subset
    K(a')$, then $K(f([a])) = K(a) \subset K(a') = K(f([a']))$, i.e.  $f([a])
    \anlt_K f([a'])$. By the assumption on $K'$ in the statement of the lemma,
    this means $f([a]) \anle_{K'} f([a'])$, and so
    \[
        K''(a) = K'(f([a])) \subseteq K'(f([a'])) = K''(a')
    \]
    In either case $K''(a) \subseteq K''(a')$, i.e. $a \anle_{K''} a'$. Since
    $a, a'$ were arbitrary, this shows that ${\anle_K} \subseteq {\anle_{K''}}$
    as required.
\end{proof}

The pieces are now in place to prove \Cref{result:chainmin_mon_compatibility}

\begin{proof}[Proof of \Cref{result:chainmin_mon_compatibility}]
    For any tournament $K$, write
    \[
        \minchmon{K} = \{
            K' \in \minch{K} \mid {\anle_K} \subseteq {\anle_{K'}}
        \}
    \]
    By \Cref{result:chainmin_mon_extend_strict_part} and
    \Cref{result:chainmin_mon_extend_full}, $\minchmon{K}$ is non-empty. Let
    $\ll$ be any total order on the set $\K$ of all tournaments.
    Define a function $\alpha: \K \to \K$ by
    \[
        \alpha(K) = \min(\minchmon{K}, {\ll}) \in \minchmon{K}
    \]
    Note that the minimum is unique since ${\ll}$ is a total order. Defining an
    operator $\phi$ by $\phi(K) = ({\anle_{\alpha(K)}}, {\bnle_{\alpha(K)}})$,
    we see that $\phi$ satisfies \axiomref{chain-min} and \axiomref{mon}, as
    required.
\end{proof}

\subsection{Proof of \Cref{prop:matchpref_weightings}}

The following preliminary result is required.

\begin{lemma}
   \label{result:powertwo}

   Let $k$ and $l$ be integers with $1 \le k \le l$. Then
   \[ \sum_{i=k}^{l}{2^{-i}} < 2^{-(k - 1)} \]

\end{lemma}

\begin{proof}

    This follows from the formula for the sum of a finite geometric
    series:
    \[
        \sum_{i=0}^{n-1}{r^i} = \frac{1-r^n}{1-r}
    \]
    which holds for all $r \ne 1$. In this case we have
    \begin{align*}
        \sum_{i=k}^{l}{2^{-i}}
        &= \sum_{i=0}^{l}{2^{-i}} - \sum_{i=0}^{k-1}{2^{-i}} \\
        &= \sum_{i=0}^{l}{\left(\frac{1}{2}\right)^i}
           -
           \sum_{i=0}^{k - 1}{\left(\frac{1}{2}\right)^i} \\
        &= \frac{
               1 - \left(\frac{1}{2}\right)^{l+1}
           }{
               1 - \left(\frac{1}{2}\right)
           }
           -
           \frac{
               1 - \left(\frac{1}{2}\right)^k
           }{
               1 - \left(\frac{1}{2}\right)
           } \\
        &= 2 \left(
           2^{-k} - 2^{-(l+1)}
        \right) \\
        &= 2^{-(k-1)} - \underbrace{2^{-l}}_{> 0} \\
        &< 2^{-(k-1)}
    \end{align*}
    as required.
\end{proof}

\begin{proof}[Proof of \Cref{prop:matchpref_weightings}]

    Let $\trianglelefteq$ be a total order on $\N \times \N$ and let $m, n \in
    \N$. For $a \in [m]$ and $b \in [n]$, write
    \[
        p(a,b)
        =
        1 + |\{(a',b') \in [m] \times [n] : (a',b') \vartriangleleft (a,b)\}|
    \]
    for the `position' of $(a,b)$ in ${\trianglelefteq} \rs ([m] \times [n])$
    (where 1 corresponds to the minimal pair). Define $w$ by
    \[
        w(a,b) = 1 + 2^{-p(a,b)}
    \]
    If we abuse notation slightly and view $w$ as an $m \times n$ matrix, we
    have, by construction, $\vect_{\trianglelefteq}(w) = (1 +
    2^{-1},\ldots,1 + 2^{-mn})$. Noting that $|K_{ab} - K'_{ab}| = [K \oplus
    K']_{ab}$ for any tournaments $K, K'$, and letting $\dotprod$ denote the
    dot product, it is easy to see that
    \begin{align*}
        d_w(K, K')
        &= \vect_{\trianglelefteq}(w)
            \dotprod \vect_{\trianglelefteq}(K \oplus K') \\
        &= (1 + 2^{-1},\ldots,1 + 2^{-mn})
           \dotprod
           \vect_{\trianglelefteq}(K \oplus K') \\
        &= d(K, K') + \bm{x} \dotprod \vect_{\trianglelefteq}(K \oplus K')
    \end{align*}
    where $\bm{x} = (2^{-1},\ldots,2^{-mn})$ and $d(K, K')$ is the
    unweighted Hamming distance. In particular, since $\bm{x}$ and
    $\vect_{\trianglelefteq}(K \oplus K')$ are non-negative, we have $d_w(K, K')
    \ge d(K, K')$.

    Now, we will show that for any $m \times n$ tournament $K$ and $K' \in
    \ch_{m,n}$ with $K' \ne \alpha_{\trianglelefteq}(K)$ we have $d_w(K,
    \alpha_{\trianglelefteq}(K)) < d_w(K, K')$. Since
    $\alpha_{\trianglelefteq}(K) \in \minch{K} \subseteq \ch_{m,n}$ by
    definition, this will show that $\alpha_{\trianglelefteq}(K)$ is the
    unique minimum in \Cref{eqn:matchpref_argmin}, as required.

    So, let $K$ be an $m \times n$ tournament and $K' \in \ch_{m,n}$. To ease
    notation, write $v = \vect_{\trianglelefteq}(K \oplus
    \alpha_{\trianglelefteq}(K))$ and $v' = \vect_{\trianglelefteq}(K \oplus K')$.
    There are two cases.

    \textbf{Case 1:} $K' \not\in \minch{K}$. In this case we have $d(K, K')
    \ge \mindist{K} + 1$, and
    \begin{align*}
       d_w(K, \alpha_{\trianglelefteq}(K))
       &=
           \underbrace{d(K, \alpha_{\trianglelefteq}(K))}_{= \mindist{K}}
           +
           \bm{x} \dotprod v \\
       &= \mindist{K} + \sum_{i=1}^{mn}{2^{-i} \cdot \underbrace{v_i}_{\le 1}}
       \\
       &\le \mindist{K} + \underbrace{\sum_{i=1}^{mn}{2^{-i}}}_{< 2^{-0} = 1}
       \\
       &< \mindist{K} + 1 \\
       &\le d(K, K') \\
       &\le d_w(K, K')
    \end{align*}
    where \Cref{result:powertwo} was applied in the 4th step. This shows $d_w(K,
    \alpha_{\trianglelefteq}(K)) < d_w(K, K')$, as required.

    \textbf{Case 2:} $K \in \minch{K}$. In this case we have
    \begin{align*}
       d(K, \alpha_{\trianglelefteq}(K)) - d(K, K')
       &= (\mindist{K} + \bm{x} \dotprod v)
          -
          (\mindist{K} + \bm{x} \dotprod v') \\
       &= \bm{x} \dotprod (v - v')
    \end{align*}

    Now, since $K' \in \minch{K}$, $v'$ appears as one of the vectors over
    which the $\argmin$ is taken in
    \Cref{eqn:match_preference_alpha_definition}. By definition of
    $\alpha_{\trianglelefteq}$ we therefore know that $v$ strictly precedes
    $v'$ with respect to the lexicographic order on $\{0,1\}^{mn}$.
    Consequently there is $j \ge 1$ such that $v_i = v'_i$ for $i < j$ and $v_j
    < v'_j$. That is, $v_j = 0$ and $v'_j = 1$. This means
    \begin{align*}
       d(K, \alpha_{\trianglelefteq}(K)) - d(K, K')
       &= \bm{x} \dotprod (v - v') \\
       &= \sum_{i=1}^{mn}{2^{-i}(v_i - v'_i)} \\
       &= \sum_{i=1}^{j-1}{2^{-i}\underbrace{(v_i - v'_i)}_{=0}}
          +
          \sum_{i=j}^{mn}{2^{-i}(v_i - v'_i)} \\
       &= 2^{-j}\underbrace{(v_j - v'_j)}_{=-1}
          + \sum_{i=j + 1}^{mn}{
               2^{-i}
               \underbrace{(v_i - v'_i)}_{\le 1}
            } \\
       &\le -2^{-j}
            + \sum_{i=j+1}^{mn}{
                  2^{-i}
              } \\
       &< -2^{-j} + 2^{-j} \\
       &= 0
    \end{align*}
    where \Cref{result:powertwo} was applied in the second to last step.
    Again, this shows $d_w(K, \alpha_{\trianglelefteq}(K)) < d_w(K, K')$, and
    the proof is complete.
\end{proof}

\subsection{Proof of \Cref{result:chain_def_ranks_characterisation}}

\begin{proof}

    First we set up some notation. For a total preorder $\preceq$ on a set $Z$
    and $z \in Z$, write $[z]_{\preceq}$ for the rank of ${\preceq}$ containing
    $z$, i.e. the equivalence class of $z$ in the symmetric closure of
    ${\preceq}$:
    \[
        [z]_{\preceq}
        = \{z' \in Z \mid z \preceq z' \text{ and } z' \preceq z\}
    \]
    Also note that $\preceq$ can be extended to a total order on the ranks by
    setting $[z]_{\preceq} \le [z']_{\preceq}$ iff $z \preceq z'$.

    (⇒) We start with the `only if' statement of the theorem. Suppose
    $\phi$ satisfies \axiomref{chain-def}, and let $K$ be a tournament. We need
    to show that $|\ranks{\ale_K^\phi} - \ranks{\ble_K^\phi}| \le 1$.

    By chain-definability, there is $K'$ with the chain property such that $a
    \ale_K^\phi a'$ iff $K'(a) \subseteq K'(a')$ and $b \ble_K^\phi b'$ iff
    $(K')^{-1}(b) \supseteq (K')^{-1}(b')$. Write
    \[ \mathcal{X} = \{ [a]_{\ale_K^\phi} \mid a \in A, K'(a) \ne \emptyset \} \]
    \[ \mathcal{Y} = \{ [b]_{\ble_K^\phi} \mid b \in B, (K')^{-1}(b) \ne \emptyset \} \]
    for the set of ranks in each of the two orders, excluding those who have
    empty neighbourhoods in $K'$. Note that $[a]_{\ale_K^\phi} =
    [a']_{\ale_K^\phi}$ if and only if $K'(a) = K'(a')$ (and similar for $B$).

    We will show that $|\mathcal{X}| = |\mathcal{Y}|$. Enumerate $\mathcal{X} =
    \{X_1,\ldots,X_s\}$ and $\mathcal{Y} = \{Y_1,\ldots,Y_t\}$, ordered such
    that $X_1 < \cdots < X_s$ and $Y_1 < \cdots < Y_t$. First we show
    $|\mathcal{X}| \le |\mathcal{Y}|$.

    For each $1 \le i \le s$, the $a_i$ be an arbitrary element of $X_i$. Then
    $a_1 \alt_K^\phi \cdots \alt_K^\phi a_s$, so $\emptyset \subset K'(a_1)
    \subset \cdots \subset K'(a_s)$. Since these inclusions are strict, we can
    choose $b_1,\ldots,b_s \in B$ such that $b_1 \in K'(a_1)$ and $b_{i+1} \in
    K'(a_{i+1}) \setminus K'(a_i)$ for $1 \le i < s$.

    It follows that $a_i \in (K')^{-1}(b_i) \setminus (K')^{-1}(b_{i+1})$, and
    thus $(K')^{-1}(b_i) \not\subseteq (K')^{-1}(b_{i+1})$. Since $K'$ has
    the chain property, this means $(K')^{-1}(b_{i+1}) \subset
    (K')^{-1}(b_i)$, i.e. $b_i \blt_K^\phi b_{i+1}$.

    We now have $b_1 \blt_K^\phi \cdots \blt_K^\phi b_s$; a chain of $s$ strict
    inequalities in $\ble_K^\phi$. The corresponding ranks $[b_1], \ldots,
    [b_s]$ are all distinct and lie inside $\mathcal{Y}$. But now we have found
    $s = |\mathcal{X}|$ distinct elements of $\mathcal{Y}$, so $|\mathcal{X}|
    \le |\mathcal{Y}|$ as promised.

    Repeating this argument with the roles of $\mathcal{X}$ and $\mathcal{Y}$
    interchanged, we find that $|\mathcal{Y}| \le |\mathcal{X}|$ also, and
    therefore $|\mathcal{X}| = |\mathcal{Y}|$.

    To conclude, note that $\ranks{\ale_K^\phi} \in \{|\mathcal{X}|,
    |\mathcal{X}| + 1\}$, since there can exist at most one rank which was
    excluded from $\mathcal{X}$ (namely, those $a \in A$ with $K'(a) =
    \emptyset$). For identical reasons, $\ranks{\ble_K^\phi} \in
    \{|\mathcal{Y}|, |\mathcal{Y}| + 1\}$. Since $|\mathcal{X}| =
    |\mathcal{Y}|$, it is clear that $\ranks{\ale_K^\phi}$ and
    $\ranks{\ble_K^\phi}$ can differ by at most one, as required.

    (⇐) Now we prove the `if' statement. Let $K$ be a tournament. We have
    $|\ranks{\ale_K^\phi} - \ranks{\ble_K^\phi}| \le 1$, and must show there is
    tournament $K'$ with the chain property such that $\phi(K) = ({\anle_{K'}},
    {\bnle_{K'}})$.

    Let $X_1 < \cdots < X_s$ and $Y_1 < \cdots < Y_t$ be the ranks of
    $\ale_K^\phi$ and $\ble_K^\phi$ respectively. By hypothesis $|s - t| \le
    1$. Define $g: \{1,\ldots,s\} \to \{0,\ldots,t\}$ by
    \[
        g(i) = \begin{cases}
            i,& s \in \{t-1, t\} \\
            i - 1,& s = t + 1
        \end{cases}
    \]
    Not that the two cases above cover all possibilities, since $|s - t| \le
    1$. For $i \in [s]$, write
    \[
        N_i = \bigcup_{0 \le j \le g(i)}{Y_j}
    \]
    where $Y_0 := \emptyset$. Note that $g(i+1) = g(i) + 1$, and consequently
    \[
        N_{i+1}
        = \bigcup_{j \le g(i) + 1}{Y_j}
        = N_i \cup Y_{g(i) + 1}
        = N_i \cup Y_{g(i + 1)}
    \]
    Since $g(i+1) > 0$ we have $Y_{g(i+1)} \ne \emptyset$, and thus $N_{i+1}
    \supset N_i$ for all $i < s$.

    Now, for any $a \in A$, let $p(a) \in [s]$ be the unique integer such that
    $a \in X_{p(a)}$; such $p(a)$ always exists since $\{X_1,\ldots,X_s\}$ is a
    partition of $A$. Note that due to the assumption on the ordering of the
    $X_i$, we have $a \ale_K^\phi a'$ if and only if $p(a) \le p(a')$.

    Let $K'$ be the unique tournament such that $K'(a) = N_{p(a)}$ for each $a
    \in A$. Since $N_1 \subset \cdots \subset N_p$, we have
    \begin{equation}
        \label{eqn:ale_k_phi_iff_anle_kprime}
        \begin{aligned}
            a \ale_K^\phi a'
            &\iff p(a) \le p(a') \\
            &\iff N_{p(a)} \subseteq N_{p(a')} \\
            &\iff K'(a) \subseteq K'(a') \\
            &\iff a \anle_{K'} a'
        \end{aligned}
    \end{equation}
    i.e. ${\ale_K^\phi} = {\anle_{K'}}$. Since ${\ale_K^\phi}$ is a total
    preorder, this shows that $K'$ has the chain property.

    It only remains to show that ${\ble_K^\phi} = {\bnle_{K'}}$. First note
    that if $a \in X_i$ and $b \in Y_j$, the fact that $\{Y_1,\ldots,Y_t\}$ are
    disjoint implies
    \begin{align*}
        a \in (K')^{-1}(b)
        &\iff b \in K'(a) = N_i = \bigcup_{0 \le k \le g(i)}{Y_k} \\
        &\iff j \le g(i)
    \end{align*}
    Hence $(K')^{-1}(b)$ only depends on $j$: every $b \in Y_j$ shares the same
    neighbourhood $M_j$, given by
    \[
        M_j = \bigcup_{i \in [s] :\ g(i) \ge j}{X_i}
    \]
    Note that if $1 \le j < t$,
    \begin{align*}
        M_j
        &= \bigcup_{i \in [s] :\ g(i) \ge j}{X_i} \\
        &= \left(\bigcup_{i \in [s] :\ g(i) \ge j + 1}{X_i}\right)
            \cup
            \left(\bigcup_{i \in g^{-1}(j)}{X_i}\right) \\
        &= M_{j+1} \cup \bigcup_{i \in g^{-1}(j)}{X_i}
    \end{align*}
    Since $1 \le j < t$ we have
    \[
        g^{-1}(j) = \begin{cases}
            \{j\},& s \in \{t-1,t\} \\
            \{j+1\},& s = t + 1
        \end{cases}
    \]
    In particular $g^{-1}(j) \ne \emptyset$, which means $\bigcup_{i \in
    g^{-1}(j)}{X_i} \ne \emptyset$ and thus $M_j \supset M_{j+1}$
    for all $1 \le j < t$.

    Finally, since $(K')^{-1}(b) = M_j$ for $b \in Y_j$ and $M_1 \supset \cdots
    \supset M_t$, an argument almost identical to
    \labelcref{eqn:ale_k_phi_iff_anle_kprime} shows that ${\ble_K^\phi} =
    {\bnle_{K'}}$.

    We have shown that $\phi(K) = ({\anle_{K'}}, {\bnle_{K'}})$ and that $K'$
    has the chain property, and the proof is therefore complete.
\end{proof}

\subsection{Proof that the interleaving procedure eventually terminates}

\begin{proposition}
    \label{prop:interleaving_terminates}

    Let $(f,g)$ be selection functions. Fix a tournament $K$ and let $A_i,
    B_i$ ($i \ge 0$) be as in \Cref{def:interleaving}. Then there are $j, j'
    \ge 1$ such that $A_j = \emptyset$ and $B_{j'} = \emptyset$. Moreover,
    there is $t \ge 1$ such that both $A_t = B_t = \emptyset$.

\end{proposition}

\begin{proof}

    Suppose $i \ge 0$ and $A_i \ne \emptyset$. Then properties
    \labelcref{item:f_sel_1} and \labelcref{item:f_sel_2} for $f$ in
    \Cref{def:selectionfunction} imply that $\emptyset \subset f(K, A_i, B_i)
    \subseteq A_i$, and consequently $A_{i+1} = A_i \setminus f(K, A_i, B_i)
    \subset A_i$.

    Supposing that $A_j \ne \emptyset$ for all $j \ge 0$, we would have $A_0
    \supset A_1 \supset A_2 \supset \cdots$ which clearly cannot be the case
    since each $A_j$ lies inside $A$ which is a finite set. Hence there is $j
    \ge 1$ such that $A_j = \emptyset$. Moreover, since $A_j \supseteq A_{j+1}
    \supseteq A_{j+2} \supseteq \cdots$, we have $A_k = \emptyset$ for all $k
    \ge j$.

    An identical argument with $g$ shows that there is $j' \ge 1$ such that
    $B_{j'} = \emptyset$ and $B_k = \emptyset$ for all $k \ge j'$.

    Taking $t = \max\{j, j'\}$, we have $A_t = B_t = \emptyset$ as required.
\end{proof}

\subsection{Proof of \Cref{result:chaindef_iff_interleaving}}

\begin{proof}

    Throughout the proof we will refer to a pair of total preorders
    $({\ale}, {\ble})$ as `chain-definable' if there is a chain tournament $K$
    such that ${\ale} = {\anle_K}$ and ${\ble} = {\bnle_K}$.

    (⇐) First we prove the `if' direction. Let $\phi = \intop{f,g}$ be an
    interleaving operator with selection functions $(f, g)$, and fix a
    tournament $K$. We will show that $\phi(K)$ is chain-definable.

    As per \Cref{prop:interleaving_terminates}, let $j, j' \ge 1$ be the
    minimal integers such that $A_j = \emptyset$ and $B_{j'} = \emptyset$. Then
    we have $A_0 \supset \cdots \supset A_{j-1} \supset A_j = \emptyset$ and
    $B_0 \supset \cdots \supset B_{j'-1} \supset B_{j'} = \emptyset$.

    Recall that, for $a \in A$, we have by definition $r(a) = \max\{i \mid a
    \in A_i\}$, which is the unique integer such that $a \in A_{r(a)} \setminus
    A_{{r(a)}+1}$. Since $a \ale_K^\phi a'$ iff $r(a) \ge r(a')$, it follows
    that the non-empty sets $A_0 \setminus A_1, \ldots, A_{j-1} \setminus A_j$
    form the ranks of the total preorder ${\ale_K^\phi}$ (that is, the
    equivalence classes of the symmetric closure ${\aeq_K^\phi}$). Thus,
    ${\ale_K^\phi}$ has $j$ ranks. An identical argument shows that
    ${\ble_K^\phi}$ has $j'$ ranks.

    It follows from \Cref{result:chain_def_ranks_characterisation} that $\phi(K)$
    is chain-definable if and only if $|j - j'| \le 1$.  If $j = j'$ this is
    clear. Suppose $j < j'$.  Then $A_j = \emptyset$ and $B_j \ne \emptyset$.
    By property \labelcref{item:g_sel_3} for $g$ in
    \Cref{def:selectionfunction}, we have $g(K, A_j, B_j) = g(K, \emptyset,
    B_j) = B_j$. But this means $B_{j+1} = B_j \setminus g(K, A_j, B_j) = B_j
    \setminus B_j = \emptyset$.  Consequently $j' = j+1$, and $|j - j'| = |-1|
    = 1$

    If instead $j > j'$, then a similar argument using property
    \labelcref{item:f_sel_3} for $f$ in \Cref{def:selectionfunction} shows
    that $j = j' + 1$, and we have $|j - j'| = |1| = 1$.

    Hence $|j - j'| \le 1$ in all cases, and $\phi(K)$ is chain-definable as
    required.

    (⇒) Now for the `only if' direction. Suppose $\phi$ satisfies
    \axiomref{chain-def}. We will define $f, g$ such that $\phi = \intop{f,g}$.
    The idea behind the construction is straightforward: since $f$ and $g$ pick
    off the next-top-ranked $A$s and $B$s at each iteration, simply define
    $f(K, A_i, B_i)$ as the maximal elements of $A_i$ with respect to the
    existing ordering ${\ale_K^\phi}$ ($g$ will be defined similarly). The
    interleaving algorithm will then select the ranks of ${\ale_K^\phi}$ and
    ${\ble_K^\phi}$ one-by-one; the fact that $\phi(K)$ is chain-definable
    ensures that we select \emph{all} the ranks before the iterative procedure
    ends. The formal details follow.

    Fix a tournament $K$. By \Cref{result:chain_def_ranks_characterisation},
    $|\ranks{{\ale_K^\phi}} - \ranks{{\ble_K^\phi}}| \le 1$. Taking $t =
    \max\{\ranks{{\ale_K^\phi}}, \ranks{{\ble_K^\phi}}\}$, we can write $X_1,
    \ldots, X_t \subseteq A$ and $Y_1, \ldots, Y_t \subseteq B$ for the ranks
    of ${\ale_K^\phi}$ and ${\ble_K^\phi}$ respectively, possibly with $X_1 =
    \emptyset$ if $\ranks{{\ble_K^\phi}} = 1 + \ranks{{\ale_K^\phi}}$ or $Y_1 =
    \emptyset$ if $\ranks{{\ale_K^\phi}} = 1 + \ranks{{\ble_K^\phi}}$. Note
    that $X_i, Y_i \ne \emptyset$ for $i > 1$. Assume these sets are ordered
    such that $a \ale_K^\phi a'$ iff $i \le j$ whenever $a \in X_i$ and $a' \in
    X_j$ (and similar for the $Y_i$). Also note that the $X_i \cap X_j =
    \emptyset$ for $i \ne j$ (and similar for the $Y_i$).

    Now set \footnotemark{}
    \[
        f(K, A', B') = \begin{cases}
           \max(A', {\ale_K^\phi}),& B' \ne \emptyset \\
           A',& B' = \emptyset
        \end{cases}
    \]
    \[
        g(K, A', B') = \begin{cases}
            \max(B', {\ble_K^\phi}),& A' \ne \emptyset \\
            B',& A' = \emptyset
        \end{cases}
    \]

    \footnotetext{
        Here $\max(Z, {\preceq}) = \{z \in Z \mid \not\exists z' \in Z : z
        \prec z'\}$, for any set $Z$ and a total preorder ${\preceq}$ on $Z$
        (with strict part ${\prec}$).
    }

    It is not difficult to see that $f$ and $g$ satisfy the conditions of
    \Cref{def:selectionfunction} for selection functions. We claim that for
    with $A_i, B_i$ denoting the interleaving sets for $K$ and $(f, g)$, for
    all $0 \le i \le t$ we have
    \begin{equation}
        \label{eqn:ai_bi_unions}
        A_i = \bigcup_{j=1}^{t - i}{X_j},
        \quad
        B_i = \bigcup_{j=1}^{t - i}{Y_j}
        \quad
        \quad
    \end{equation}
    For $i = 0$ this is clear: since $X_1,\ldots,X_t$ contains all ranks of
    ${\ale_K^\phi}$ we have $\bigcup_{j=1}^{t-0} = X_1 \cup \cdots \cup X_t = A
    = A_0$ (and similar for $B$).

    Now suppose \labelcref{eqn:ai_bi_unions} holds for some $0 \le i < t$. We
    will show that $f(K, A_i, B_i) = X_{t-i}$ by considering three possible
    cases, at least one of which must hold.

    \textbf{Case 1:} ($A_i \ne \emptyset$, $B_i \ne \emptyset$). Here we have
    \begin{align*}
        f(K, A_i, B_i)
        &= \max(A_i, {\ale_K^\phi}) \\
        &= \max(X_1 \cup \cdots \cup X_{t-i}, {\ale_K^\phi}) \\
        &= X_{t-i}
    \end{align*}
    since the $X_j$ form (disjoint) ranks of ${\ale_K^\phi}$ with $X_{j} \alt
    X_{k}$ for $j < k$.

    \textbf{Case 2:} ($B_i = \emptyset$). Here we have
    $\bigcup_{j=1}^{t-i}{Y_j} = \emptyset$. Since $t - i \ge 1$ and $Y_j \ne
    \emptyset$ for $j > 1$, it must be the case that $t - i = 1$ and $B_i = Y_1
    = \emptyset$. Consequently by the induction hypothesis we have $A_i =
    \bigcup_{j=1}^{1}{X_j} = X_1$, and thus
    \begin{align*}
        f(K, A_i, B_i)
        &= f(K, A_i, \emptyset) \\
        &= A_i \\
        &= X_1 \\
        &= X_{t-i}
    \end{align*}

    \textbf{Case 3:} ($A_i = \emptyset$). By a similar argument as in case 2,
    we must have $t - i = 1$ and $A_i = X_1 = \emptyset$. Using the fact that
    $f(K, A_i, B_i) \subseteq A_i$ we get
    \begin{align*}
        f(K, A_i, B_i)
        &= \underbrace{f(K, \emptyset, B_i)}_{\subseteq \emptyset} \\
        &= \emptyset \\
        &= X_1 \\
        &= X_{t-i}
    \end{align*}
    We have now covered all cases, and have shown that $f(K, A_i, B_i) =
    X_{t-i}$ must hold. Consequently, using again the fact that the $X_j$ are
    disjoint,
    \begin{align*}
        A_{i+1}
        &= A_i \setminus f(K, A_i, B_i) \\
        &= \left(\bigcup_{j=1}^{t-i}{X_j}\right) \setminus X_{t-i} \\
        &= \bigcup_{j=1}^{t-(i+1)}{X_j}
    \end{align*}
    as required. By almost identical arguments we can show that $g(K, A_i, B_i)
    = Y_{t-i}$, and thus $B_{i+1} = \bigcup_{j=1}^{t-(i+1)}{Y_j}$ also. By
    induction, \labelcref{eqn:ai_bi_unions} holds for all $0 \le i \le t$.

    It remains to show that $a \ale_K^\phi a'$ iff $a \ale_K^{\intop{f,g}} a'$
    and that $b \ble_K^\phi b'$ iff $b \ble_K^{\intop{f,g}} b'$.

    For $a \in A$, let $p(a)$ be the unique integer such that $a \in X_{p(a)}$,
    i.e. $p(a)$ is the index of the rank of $a$ in the ordering
    ${\ale_K^\phi}$. Note that we have
    \[
        a \in A_i = X_1 \cup \cdots \cup X_{t-i}
        \iff
        t - i \ge p(a)
    \]
    and therefore
    \[
        r(a)
        = \max\{i \mid a \in A_i\}
        = \max\{i \mid t - i \ge p(a)\}
        = t - p(a)
    \]
    Using the fact that $X_i \alt X_j$ for $i < j$, we get
    \begin{align*}
        a \ale_K^{\intop{f,g}} a'
        &\iff r(a) \ge r(a') \\
        &\iff t - p(a) \ge t - p(a') \\
        &\iff p(a) \le p(a') \\
        &\iff a \ale_K^\phi a'
    \end{align*}
    A similar argument shows that $b \ble_K^\phi b'$ iff $b
    \ble_K^{\intop{f,g}} b'$ for any $b, b' \in B$. Since $K$ was arbitrary, we
    have shown that $\phi = \intop{f,g}$ as required.
\end{proof}

\subsection{Proof of \Cref{result:chaindef_axiom_compatibilities}}

\begin{proof}

    Since \axiomref{chain-min} implies \axiomref{chain-def},
    \Cref{result:chainmin_axiom_compatibilities} implies the existence of an
    operator with \axiomref{chain-def} and \axiomref{dual}, and an operator
    with \axiomref{chain-def} and \axiomref{mon}.  Moreover, the trivial
    operator which ranks all $A$s and $B$s equally satisfies \axiomref{anon}
    and \axiomref{IIM}. It only remains to show that there is an operator
    satisfying both \axiomref{chain-def} and \axiomref{pos-resp}.

    To that end, for any tournament $K$, define $K'$ by
    \[
        K'_{ab} = \begin{cases}
            1 ,& b \le |K(a)| \\
            0 ,& b > |K(a)|
        \end{cases}
    \]
    Note that $K'(a) = \{1,\ldots,|K(a)|\}$ for $|K(a)| > 0$. Consequently
    $K'(a) \subseteq K'(a')$ iff $|K(a)| \le |K(a)|$. We see that $K'$ has the
    chain property, and the operator $\phi$ defined by $\phi(K) =
    ({\anle_{K'}}, {\bnle_{K'}})$ satisfies \axiomref{chain-def}. In
    particular, $a \ale_K^\phi a'$ iff $|K(a)| \le |K(a')|$.

    To show \axiomref{pos-resp}, suppose $a \ale_K^\phi a'$ and $K_{a',b} = 0$
    for some $a, a' \in A$ and $b \in B$. Write $\hat{K} = K + \bm{1}_{a',b}$.

    Since $a \ale_K^\phi a'$ implies $|K(a)| \le |K(a')|$, we have
    $|\hat{K}(a')| = 1 + |K(a')| > |K(a)| = |\hat{K}(a)|$, and therefore $a
    \alt_{\hat{K}}^\phi a'$ as required for \axiomref{pos-resp}.
\end{proof}

\subsection{Proof of \Cref{result:chaindef_impossibility}}

\begin{proof}

    For contradiction, suppose there is an operator $\phi$ satisfying the
    stated axioms. Consider
    \[
        K = \left[\begin{smallmatrix}
            0 & 0 \\
            0 & 1 \\
            1 & 0 \\
            1 & 1
        \end{smallmatrix}\right]
    \]
    and two tournaments obtained by removing a single 1 entry:
    \[
        K_1 = \left[\begin{smallmatrix}
            0 & 0 \\
            0 & \bm{{\color{red}0}} \\
            1 & 0 \\
            1 & 1
        \end{smallmatrix}\right],
        \quad
        K_2 = \left[\begin{smallmatrix}
            0 & 0 \\
            0 & 1 \\
            1 & 0 \\
            1 & \bm{{\color{red}0}}
        \end{smallmatrix}\right]
    \]
    Now, \axiomref{anon} in $K_1$ gives $1 \aeq_{K_1}^\phi 2$ (e.g. take
    $\sigma = (1\ 2)$, $\pi = \text{id}_B$). In particular, $1 \ale_{K_1}^\phi
    2$, so \axiomref{pos-resp} implies $1 \alt_K^\phi 2$. A similar argument
    with $K_2$ shows that $3 \aeq_{K_2}^\phi 4$ and $3 \alt_K^\phi 4$.

    On the other hand, applying \axiomref{anon} to $K$ directly with $\sigma =
    (2\ 3)$ and $\pi = (1\ 2)$, we see that $2 \aeq_K^\phi 3$. The ranking of
    $A$ is thus fully determined as $1 \alt 2 \aeq 3 \alt 4$. In particular,
    $\ranks{\ale_K^\phi} = 3$.

    But now considering the dual tournament $\dual{K} =
    \left[\begin{smallmatrix} 1&1&0&0 \\ 1&0&1&0 \end{smallmatrix}\right]$ and
    applying permutations $\sigma = (1\ 2)$ and $\pi = (2\ 3)$, we obtain $1
    \aeq_{\dual{K}}^\phi 2$ by \axiomref{anon}, i.e. the $A$ ranking in
    $\dual{K}$ is flat. By \axiomref{dual} this implies the $B$ ranking in $K$
    is flat, i.e. $\ranks{\ble_K^\phi} = 1$. We see that $\ranks{\ale_K^\phi}$
    and $\ranks{\ble_K^\phi}$ differ by 2, contradicting \axiomref{chain-def}
    according to \Cref{result:chain_def_ranks_characterisation}.
\end{proof}

\subsection{Proof of \Cref{result:phicardint_axioms}}

We require a preliminary result providing sufficient conditions for an
interleaving operator $\intop{f,g}$ to satisfy various axioms.

\begin{lemma}
   \label{result:interleaving_suffconditions}

   Let $\phi = \phi_{f,g}^{\text{int}}$ be an interleaving operator.

   \begin{enumerate}
       \item \label{item:int_lemma_anon}

        If for any tournament $K$, $A' \subseteq A$, $B' \subseteq B$ and for
        any pair of permutations $\sigma: A \to A$ and $\pi: B \to B$ we have
        \begin{align*}
            f(\pi(\sigma(K)), \sigma(A'), \pi(B')) &= \sigma(f(K, A', B')) \\
            g(\pi(\sigma(K)), \sigma(A'), \pi(B')) &= \pi(g(K, A', B'))
        \end{align*}
        then $\phi$ satisfies \axiomref{anon}.

   \item \label{item:int_lemma_dual}

        If for any tournament $K$ and $A' \subseteq A$, $B \subseteq B$ we have
        \[ g(K, A', B') = f(\dual{K}, B', A') \] then $\phi$ satisfies
        \axiomref{dual}.

   \item \label{item:int_lemma_mon}

        If for any tournament $K$, $A' \subseteq A$, $B' \subseteq B$ and $a,
        a' \in A'$ we have
        \[
            K(a) \subseteq K(a')
            \implies
            a \not\in f(K, A', B') \text{ or } a' \in f(K, A', B')
        \]
        then $\phi$ satisfies \axiomref{mon}.

    \end{enumerate}

\end{lemma}

\begin{proof}

    We take each statement in turn.

    \begin{enumerate}

    \item
    Let $K$ be a tournament. For brevity, write $K' = \pi(\sigma(K))$.  Let us
    write $A_i, B_i$ and $A_i', B_i'$ $(i \ge 0)$ for the sets defined in
    \Cref{def:interleaving} for $K$ and $K'$ respectively.  We claim that for
    all $i \ge 0$:
    \begin{equation}
        \label{eqn:interleaving_lemma_anon}
        A_i' = \sigma(A_i),
        \quad
        B_i' = \pi(B_i)
    \end{equation}
    For $i = 0$ this is trivial since $A_0' = A = \sigma(A) = \sigma(A_0)$
    since $\sigma$ is a bijection. The fact that $B_0' = \pi(B_0)$ is shown
    similarly.

    Suppose that \labelcref{eqn:interleaving_lemma_anon} holds for some $i \ge
    0$. Then applying our assumption on $f$:
    \begin{align*}
        A_{i+1}'
        &= A_i' \setminus f(K', A_i', B_i') \\
        &= \sigma(A_i) \setminus f(K', \sigma(A_i), \pi(B_i)) \\
        &= \sigma(A_i) \setminus \sigma(f(K, A_i, B_i)) \\
        &= \sigma(A_i \setminus f(K, A_i, B_i)) \\
        &= \sigma(A_{i+1})
    \end{align*}
    (note that $\sigma(X) \setminus \sigma(Y) = \sigma(X \setminus Y)$ holds
    for any sets $X, Y$ due to injectivity of $\sigma$). Using the assumption
    on $g$ we can show that $B_{i+1}' = \pi(B_{i+1})$ in a similar manner.
    Therefore, by induction, \labelcref{eqn:interleaving_lemma_anon} holds for
    all $i \ge 0$. This means that for any $a \in A$ we have
    \[
        \sigma(a) \in A_i'
        \iff
        \sigma(a) \in \sigma(A_i)
        \iff
        a \in A_i
    \]
    and therefore, with $r_K$ and $r_{K'}$ denoting the functions $A \to \N_0$
    defined in \Cref{def:interleaving} for $K$ and $K'$ respectively,
    \begin{align*}
        r_{K'}(\sigma(a))
        &= \max\{i \mid \sigma(a) \in A_i'\} \\
        &= \max\{i \mid a \in A_i\} \\
        &= r_K(a)
    \end{align*}

    From this it easily follows that $a \ale_K^\phi a'$ iff $\sigma(a)
    \ale_{K'}^\phi \sigma(a')$, i.e.  $\phi$ satisfies \axiomref{anon}.

    \item
    Once again, fix a tournament $K$ and let $A_i, B_i$ and $A_i', B_i'$ denote
    the sets from \Cref{def:interleaving} for $K$ and $\dual{K}$ respectively.
    It is easy to show by induction that the assumption on $f$ and $g$ implies
    $A'_i = B_i$ and $B_i' = A_i$ for all $i \ge 0$ . This means that for any
    $b \in B_K$:
    \begin{align*}
        s_K(b)
        &= \max\{i \mid b \in B_i\} \\
        &= \max\{i \mid b \in A_i'\} \\
        &= r_{\dual{K}}(b)
    \end{align*}
    which implies $b \ble_K^\phi b'$ iff $b \ale_{\dual{K}}^\phi b'$, as
    required for \axiomref{dual}.

    \item
    Let $K$ be a tournament and $a, a' \in A$ such that $K(a) \subseteq K(a')$.
    We must show that $a \ale_K^\phi a'$.

    Suppose otherwise, i.e. $a' \alt_K^\phi a$. Then $r(a') > r(a)$. Note that
    by definition of $r$, we have $a \in A_{r(a)} \setminus A_{r(a) + 1} = f(K,
    A_{r(a)}, B_{r(a)})$. Since $r(a') \ge r(a) + 1$ and $A_{r(a)} \supseteq
    A_{r(a) + 1} \supseteq A_{r(a) + 2} \supseteq \cdots$, we get $a' \in
    A_{r(a) + 1} \subseteq A_{r(a)}$. In particular, $a' \notin f(K, A_{r(a)},
    B_{r(a)})$.

    Piecing this all together, we have $a, a' \in A_{r(a)}$, $K(a) \subseteq
    K(a')$, $a \in f(K, A_{r(a)}, B_{r(a)})$ and $a' \not\in f(K, A_{r(a)},
    B_{r(a)})$. But this directly contradicts our assumption on $f$, so we are
    done.

    \end{enumerate}
\end{proof}

\begin{proof}[Proof of \Cref{result:phicardint_axioms}]

    We take each axiom in turn. Let $f$ and $g$ be the selection functions
    corresponding to $\phicardint$ from \Cref{ex:cardint}.

    \axiomref{chain-def.} Since $\phicardint$ is an interleaving operator,
    \axiomref{chain-def} follows from \Cref{result:chaindef_iff_interleaving}.

    \axiomref{anon.} Let $K$ be a tournament and let $\sigma: A \to A$ and $\pi:
    B \to B$ be bijective mappings. Write $K' = \pi(\sigma(K))$. We will show
    that the conditions on $f$ and $g$ in
    \Cref{result:interleaving_suffconditions} part
    (\labelcref{item:int_lemma_anon}) are satisfied.

    Let $A' \subseteq A$ and $B' \subseteq B$. We have
    \begin{align*}
        f(K', \sigma(A'), \pi(B'))
        &= \argmax_{\hat{a} \in \sigma(A')}{|K'(\hat{a}) \cap \pi(B')|} \\
        &= \sigma(\argmax_{a \in A'}{|K'(\sigma(a)) \cap \pi(B')|})
    \end{align*}
    where we make the `substitution' $a = \sigma^{-1}(\hat{a})$. Using the
    defintion of $K' = \pi(\sigma(K))$ it is easily seen that $K'(\sigma(a)) =
    \pi(K(a))$. Also, since $\pi$ is a bijection we have $\pi(X) \cap \pi(Y) =
    \pi(X \cap Y)$ for any sets $X$ and $Y$, and $|\pi(X)| = |X|$. Thus
    \begin{align*}
        f(K', \sigma(A'), \pi(B'))
        &= \sigma(\argmax_{a \in A'}{|K'(\sigma(a)) \cap \pi(B')|}) \\
        &= \sigma(\argmax_{a \in A'}{|\pi(K(a)) \cap \pi(B')|}) \\
        &= \sigma(\argmax_{a \in A'}{|\pi(K(a) \cap B')|}) \\
        &= \sigma(\argmax_{a \in A'}{|K(a) \cap B'|}) \\
        &= \sigma(f(K, A', B'))
    \end{align*}
    as required. The result for $g$ follows by a near-identical argument. Thus
    $\phicardint$ satisfies \axiomref{anon} by
    \Cref{result:interleaving_suffconditions} part
    (\labelcref{item:int_lemma_anon}).

    \axiomref{dual.} Fix a tournament $K$ and let $A' \subseteq A$, $B'
    \subseteq B$. Note that for $b \in B'$ we have
    \begin{align*}
        |K^{-1}(b) \cap A'|
        &= |(A \setminus \dual{K}(b)) \cap A'| \\
        &= |A' \setminus \dual{K}(b)| \\
        &= |A'| - |\dual{K}(b) \cap A'|
    \end{align*}
    Consequently
    \begin{align*}
        g(K, A', B')
        &= \argmin_{b \in B'}{|K^{-1}(b) \cap A'|} \\
        &= \argmin_{b \in B'}{\left(|A'| - |\dual{K}(b) \cap A'|\right)} \\
        &= \argmax_{b \in B'}{|\dual{K}(b) \cap A'|} \\
        &= f(\dual{K}, B', A')
    \end{align*}
    and, by \Cref{result:interleaving_suffconditions} part
    (\labelcref{item:int_lemma_dual}), $\phicardint$ satisfies \axiomref{dual}.

    \axiomref{mon.} Once again, we use
    \Cref{result:interleaving_suffconditions}.  Let $K$ be a tournament and $A'
    \subseteq A$, $B' \subseteq B$. Suppose $a, a' \in A'$ with $K(a) \subseteq
    K(a')$. We need to show that either $a \not\in f(K, A', B')$ or $a' \in
    f(K, A', B')$

    Suppose $a \in f(K, A', B')$. Then $a \in \argmax_{\hat{a} \in
    A'}{|K(\hat{a}) \cap B'|}$, so $|K(a) \cap B'| \ge |K(a') \cap B'|$. On the
    other hand $K(a) \cap B' \subseteq K(a') \cap B'$, so $|K(a) \cap B'| \le
    |K(a') \cap B'|$.  Consequently $|K(a) \cap B'| = |K(a') \cap B'|$, and so
    $a' \in f(K, A', B')$. This shows the property required by
    \Cref{result:interleaving_suffconditions} part
    (\labelcref{item:int_lemma_mon}) is satisfied, and thus $\phicardint$
    satisfies \axiomref{mon}.

    \axiomref{pos-resp.} We have show that $\phicardint$ satisfies
    \axiomref{chain-def}, \axiomref{anon} and \axiomref{dual}; due to
    impossibility result of \Cref{result:chaindef_impossibility}, $\phicardint$
    cannot satisfy \axiomref{pos-resp}.

    \axiomref{IIM.} Write
    \[
        K_1 = \begin{bmatrix}
            1 & 0 & 0 \\
            0 & 1 & 0 \\
            0 & 1 & 1
        \end{bmatrix}
        , \quad
        K_2 = \begin{bmatrix}
            1 & 0 & 0 \\
            0 & 1 & 0 \\
            1 & 0 & 1
        \end{bmatrix}
    \]
    Note that the first and second rows of each tournament are identical, so
    \axiomref{IIM} would imply $1 \ale_{K_1}^\phicardint 2$ iff $1
    \ale_{K_2}^\phicardint 2$.  However, it is easily verified that $1
    \alt_{K_1}^\phicardint 2$ whereas $2 \alt_{K_2}^\phicardint 1$. Therefore
    $\phicardint$ does not satisfy \axiomref{IIM}.
\end{proof}

\end{document}